\documentclass[12pt]{article}
\usepackage{amsmath, amsthm}
\usepackage{graphicx}
\usepackage{natbib}
\usepackage{url} 
\usepackage{amssymb}
\usepackage{amsbsy}
\usepackage{setspace}
\usepackage{enumerate}
\usepackage{algorithm}
\usepackage{multirow}
\usepackage{mathtools}
\usepackage{booktabs}
\usepackage{array}             
\usepackage{diagbox}
\usepackage[noend]{algpseudocode}
\usepackage[dvipsnames]{xcolor}

\usepackage{xr}
\makeatletter
\newcommand*{\addFileDependency}[1]{
\typeout{(#1)}
\@addtofilelist{#1}
\IfFileExists{#1}{}{\typeout{No file #1.}}
}\makeatother

\newcommand{\Y}{{\mbox{\boldmath $Y$}}}

\newcommand{\mbf}[1]{\mbox{\boldmath${#1}$}}
\newcommand{\X}{{\mbox{\boldmath $X$}}}

\newcommand{\Z}{{\mbox{\boldmath $Z$}}}
\newcommand{\U}{{\mbox{\boldmath $U$}}}
\newcommand{\W}{{\mbox{\boldmath $W$}}}
\newcommand{\V}{{\mbox{\boldmath $V$}}}

\newcommand{\bG}{{\mbox{\boldmath $G$}}}

\newcommand{\I}{{\mbox{\boldmath $I$}}}

\newcommand{\re}{{\mbox{\boldmath $b$}}}

\newcommand{\be}{\mbf{\beta}}

\newcommand{\balpha}{{\mbox{\boldmath $\alpha$}}}

\newcommand{\bgamma}{{\mbox{\boldmath $\gamma$}}}
\newcommand{\bomega}{{\mbox{\boldmath $\omega$}}}
\newcommand{\btau}{{\mbox{\boldmath $\tau$}}}

\newcommand{\btheta}{{\mbox{\boldmath $\theta$}}}

\newcommand{\bSigma}{{\mbox{\boldmath $\Sigma$}}}
\newcommand{\bPsi}{{\mbox{\boldmath $\Psi$}}}
\newcommand{\bphi}{{\mbox{\boldmath $\phi$}}}
\newcommand{\bfeta}{{\mbox{\boldmath $\xi$}}}

\newcommand{\bvarphi}{{\mbox{\boldmath $\bvarphi$}}}

\newcommand{\blind}{0}

\newtheorem{proposition}{Proposition}

\begin{document}

\def\spacingset#1{\renewcommand{\baselinestretch}%
{#1}\small\normalsize} \spacingset{1}


\if0\blind
{
  \title{\bf Sparse high-dimensional linear mixed modeling with a partitioned empirical Bayes ECM algorithm}
  \author{Anja Zgodic$^a$, Ray Bai$^b$, Jiajia Zhang$^a$, Peter Olejua$^a$, \\ and Alexander C. McLain$^a$\hspace{.2cm}\\
    $^a$Department of Epidemiology and Biostatistics, University of South Carolina \\
    $^b$Department of Statistics, University of South Carolina}
  \maketitle
} \fi

\if1\blind
{
  \bigskip
  \bigskip
  \bigskip
  \begin{center}
    {\LARGE\bf Title}
\end{center}
  \medskip
} \fi

\bigskip
\begin{abstract}
High-dimensional longitudinal data is increasingly used in a wide range of scientific studies. \textcolor{black}{ To properly account for dependence between longitudinal observations, statistical methods for high-dimensional linear mixed models (LMMs) have been developed. However, few packages implementing these high-dimensional LMMs are available in the statistical software \texttt{R}. Additionally, some packages suffer from scalability issues.} This work presents an efficient and accurate Bayesian framework for high-dimensional LMMs. We use empirical Bayes estimators of hyperparameters for increased flexibility and an Expectation-Conditional-Minimization (ECM) algorithm for computationally efficient maximum a posteriori probability (MAP) estimation of parameters. The novelty of the approach lies in its partitioning and parameter expansion as well as its fast and scalable computation. We illustrate Linear Mixed Modeling with PaRtitiOned empirical Bayes ECM (LMM-PROBE) in simulation studies evaluating fixed and random effects estimation along with computation time. A real-world example is provided using data from a study of lupus in children, where we identify genes and clinical factors associated with a new lupus biomarker and predict the biomarker over time. \textcolor{black}{Supplementary materials are available online}. 
\end{abstract}

\noindent%
\vfill

\newpage
\spacingset{1.75} 

\section{Introduction} \label{sec.intro}

\textcolor{black}{Research on longitudinal high-dimensional data or grouped (clustered) high-dimensional data has recently attracted greater interest.} For example, some genetic studies gather gene expression levels for an individual on multiple occasions over time \citep{Banchereau2016}. Other ongoing studies -- like the UK Biobank and the Adolescent Brain Cognitive Development Study -- collect high-dimensional genetic/imaging information longitudinally to learn how individual changes in these markers are related to outcomes \citep{Cole2020, Sara2022}. Such data usually violates the traditional linear regression assumption that observations are independently and identically distributed. Data analysis should account for the dependence between observations belonging to the same individual. 

\textcolor{black}{For the low dimensional setting where $n \gg p$, linear mixed models (LMMs) are available for handling dependent data structures. For high dimensional settings where $n \ll p$, LMMs have been extended through four types of procedures: information criteria \citep[][among others]{Ariyo2020, Craiu2018}; shrinkage methods \citep[][among others]{Bondell2010, Fan2012Aim3, Groll2014, Ibrahim2011, Opoku2021, Sholokhov2024}; the Fence method \citep[][among others]{Jiang2008}; and Bayesian approaches \citep[][among others]{Degani2022, Kinney2007, Zhou2013}.} 
\textcolor{black}{Of these procedures, a small subset are implemented in readily available packages for the commonly used statistical software \texttt{R}}. 

\textcolor{black}{The shrinkage-based high-dimensional LMM approaches} focus on scenarios where either only the fixed effects or both the fixed and random effects are high-dimensional. None of the procedures that perform variable selection on both fixed and random \textcolor{black}{effects \citep[c.f., ][among others]{Bondell2010, Chen2003, Ibrahim2011, Li2021}\nocite{Fan2012Aim3, Komarek2008, Peng2012} have publicly available \texttt{R} packages. Note that we differentiate this research from genetic studies subjected to population stratification \citep{Reisetter2021}, }where LMMs are commonly used to adjust for unobserved environmental confounding \citep{Bhatnagar2020, Rakitsch2012}.  

For scenarios with high-dimensional fixed effects and low-dimensional random effects, \cite{Schelldorfer2011Aim3} introduced the first usage of penalized LMMs through a LASSO-penalized maximum likelihood estimation procedure. This procedure is implemented in the \verb|lmmlasso| \verb|R| package, but performs $p \times p$ matrix operations that are computationally expensive for large $p$. Alternatively, \cite{Rohart2014} penalize the fixed effects of the complete-data likelihood, treat random effects as missing data, and perform estimation using an Expectation-Conditional-Maximization \citep[ECM][]{Meng1993} algorithm in the \verb|MMS| \verb|R| package. This approach is doubly iterative, where fixed effects are updated with LASSO within each ECM iteration. Other software for high-dimensional clustered data includes \verb|PGEE| \citep{Wang2012}, which implements penalized generalized estimating equations (PGEE) but requires computations with $p \times p$ matrices. 

We use the aforementioned packages with a toy example to illustrate their lack of scalability to (ultra) high-dimensional longitudinal datasets. We leverage the popular riboflavin dataset \citep{Buhlmann2014}, which contains $p = 4088$ predictors, and model the production rate of riboflavin over time using the \verb|lmmlasso|, \verb|MMS|, and \verb|PGEE| packages. For one iteration of Cross-Validation (CV), \verb|lmmlasso| and \verb|PGEE| did not converge within eight hours of computation, while \verb|MMS| with LASSO took approximately 19 minutes per CV iteration (with Elastic-Net $\sim 25$ minutes, see Figure \ref{fig.time.example2}). This demonstrates that the scalability of \verb|lmmlasso|, \verb|MMS|, and \verb|PGEE| is limited to a fraction of what is currently considered `high-dimensional' in the statistical literature -- $p$ growing at a rate of $\mathcal{O}(n)$ \citep{Fan2008}. We further detail the computational complexity of these methods in Section \ref{sec.comp.anal}. 

These limitations in scalability highlight a lack of efficient software for high-dimensional LMMs. To fill this crucial gap, we propose an estimation procedure named Linear Mixed Modeling via a PaRtitiOned empirical Bayes ECM (LMM-PROBE), a method for sparse high-dimensional linear mixed modeling. In contrast to the above methods, analyzing the riboflavin dataset with LMM-PROBE took 10 seconds per CV iteration. LMM-PROBE performs Bayesian variable selection on the fixed effects and provides maximum a posteriori probability (MAP) estimates of the variable selection indicators, fixed effects, random effects, and variance components while circumventing $p \times p$ matrix computations and doubly iterative optimization. The approach is based on the recently proposed PROBE algorithm \citep{McLain2022}, which focuses on MAP estimation of the regression coefficients with minimal prior assumptions. PROBE utilizes a Parameter-Expanded ECM \citep[PX-ECM, ][]{MenRub92,Liuetal98} algorithm for variable selection and estimation. The E-step is motivated by the two-group approach to multiple testing \citep{Efron2001,Sun2007}, and is facilitated with empirical Bayes estimates of hyperparameters. 

The novelty of LMM-PROBE lies in its partitioning and parameter expansion set in a Bayesian framework, where we update fixed effects in closed form in each iteration. The contributions of the LMM-PROBE method are i) a new framework for (ultra) high-dimensional linear mixed effects regression with variable selection, ii) a Bayesian approach with increased flexibility (minimal assumptions) through the use of empirical Bayes estimators for hyperparameters, iii) a competitively fast and effective estimation procedure that scales linearly in $p$ and $n$ through a PX-ECM algorithm, and iv) an \verb|R| package \verb|lmmprobe| which implements the proposed method \citep{ZgoMcL23}. The remainder of this article is structured as follows. In Section \ref{sec.methods}, we introduce our proposed LMM-PROBE framework. In Section \ref{sec.multicycle}, we describe the PX-ECM algorithm and its computational complexity. Section \ref{sec.sim} shows numerical study results, while Section \ref{sec.data} applies the method to a real dataset on lupus in children. \textcolor{black}{Finally, Section \ref{sec.disc} provides a deeper contrast between approaches as well as a brief discussion.  }

\section{Methods}\label{sec.methods}

\subsection{Bayesian linear mixed model setup}\label{sec.blmm.setup}

Consider a nested data structure with clusters $i$ (units), $i = 1,2,\ldots,N$, each with $n_i$ observations. For cluster $i$, let $\Y_i \in \mathbb{R}^{n_i}$ denote the response vector, $\X_i \in \mathbb{R}^{n_i \times p}$ the \textit{sparse} fixed effects design matrix, and $\V_i \in \mathbb{R}^{n_i \times r}$ the \textit{non-sparse} random effects design matrix. For notational convenience and without loss of generality, we denote the \textit{non-sparse} fixed effect design matrix as $\V_i$. Our method can include other `adjustment' variables (fixed effects) that are not subject to the sparsity assumption, nor are random effects. 

For cluster $i$ our linear mixed effects model is given by 
\begin{equation}\label{eq.lmm}
\Y_{i} = \X_{i} (\bgamma \be)  + \V_{i} \mbf{\omega}  +  \V_{i} \re_i + \mbf{\epsilon}_{i},
\end{equation}
where $\bgamma\be$ denotes the Hadamard product of $\bgamma$ and $\be$, $\be \in \mathbb{R}^{p}$, $\bgamma \in \{0,1\}^p$, and $\mbf{\omega} \in \mathbb{R}^{r}$. $\be$ and $\mbf{\omega}$ are the sparse and non-sparse fixed effects, respectively, $\re_i \in \mathbb{R}^{r}$ are the random effects with $\re_i \sim N(\mbf{0}, \mbf{G})$, $\mbf{\epsilon}_{i}$ is the error term with $E(\mbf{\epsilon}_{i}) = \mbf{0}$,  $Var(\mbf{\epsilon}_{i}) = \sigma^2\mbf{I}_{n_i}$, where $\mbf{I}_{m}$ denotes an $m \times m$ identity matrix, and $Cov(\mbf{\epsilon}_{i}, \mbf{\epsilon}_{i'}) = 0$ . \textcolor{black}{At a minimum, $\V_i$ includes an intercept variable ($\V_{i,1} = \mbf{1}'$). }
Further, define $\bSigma_i = \V_{i}' \mbf{G} \V_{i} +  \sigma^2\mbf{I}_{n_i}$ and let $\Y$, $\re$, $\mbf{\epsilon}$, $\X$, $\V$ be obtained by vertically stacking vectors $\Y_i$, $\re_i$, $ \mbf{\epsilon}_i$ and matrices $\X_i$, $\V_i$ for each cluster $i$, respectively. Additionally, let $\mbf{\mathcal{V}}$, $\bSigma$, and $\mbf{\mathcal{G}}$ represent block diagonal matrices, with the $i$th block being $\V_i$, $\bSigma_i$, and \bG, respectively. Adding a Gaussian assumption on $\mbf{\epsilon}_i$ yields
\begin{eqnarray*}
\Y | \re \sim N\left\{\X' (\bgamma \be) + \mbf{\V}' \bomega + \mbf{\mathcal{V}}' \re ,  \sigma^2\I \right\} \mbox{  and  } \Y \sim N\left\{\X' (\bgamma \be) + \mbf{\V}' \bomega, \bSigma \right\}, 
\end{eqnarray*}
where $\I = \I_{M}$ an $M \times M$ identity matrix and $M = \sum_{i} n_i$.

Our Bayesian linear mixed model uses the following priors for $(\bG^{-1}, \sigma^2, \be, \bgamma)$:
\begin{eqnarray*}
\mbf{G}^{-1} &\sim& \mathcal{W}\left\{(\rho \mbf{C})^{-1}, \rho \right\}, \\
\sigma^2  &\sim& f_{\sigma^2} \propto (\sigma^2)^{-1}, \\
p(\be ) &=& \prod_{k=1}^p f_\beta(\beta_k), \\
p(\bgamma|\pi) &=& \pi^{p-|\gamma| } (1-\pi)^{|\gamma|},
\end{eqnarray*}
where $|\bgamma|$ is the sum of $\gamma_k$ values, and $\mathcal{W}$ refers to the Wishart distribution. 
To avoid specifying informative hyperpriors, we leave hyperparameters $\pi$ and $f_\beta$ unspecified and estimate them using plug-in empirical Bayes estimators. Throughout, we use $\mbf{C} =0 \I_r$ and $\rho = r + 1$.

Given $\sigma^2$, $\bgamma$, and $\re$, standard results can be leveraged to show that, when breaking $\X$ down into $(\X_{\gamma}, \X_{\bar \gamma})'$, for predictors where  $\gamma_k=1$ and $\gamma_k=0$, respectively, and $\be$ into $(\be_{\gamma} \ \be_{\bar \gamma})^\prime$, the posterior distribution of $\bgamma \be$ is 
\begin{equation*}\label{eq.post.beast}
\be_\gamma |(\Y,\sigma^2,\bgamma) \sim N\left\{  \hat{\be}_{\gamma} , \sigma^2(\X_\gamma'\X_\gamma)^{-1} \right\}
\end{equation*}
for predictors $\X_{\gamma}$, where $\hat \be_\gamma = (\X_\gamma'\X_\gamma)^{-1}\X'_\gamma(\Y - \mbf{\V}' \bomega - \mbf{\mathcal{V}}' \re)$, while for predictors in $\X_{\bar \gamma}$ the posterior of $\bgamma \be$ is $\delta_0(\cdot)$, a point mass at zero. The conditional distribution for $\sigma^2$ is inverse-gamma ($IG$) with parameters $\left(a, d\right)$, 
\begin{equation*}\label{eq.post.sigma2}
\sigma^2 |(\Y, a, d, \re, \bgamma, \be, \bomega) \sim 
\\ IG\left\{ \underbrace{\frac{M-2}{2}}_{a} , \underbrace{ \frac{1}{2} \lVert \Y - \X' \hat \be_\gamma - \mbf{\V}' \bomega - \mbf{\mathcal{V}}' \re \rVert_2^{2} }_{d} \right\},
\end{equation*}
where $\lVert \boldsymbol{x} \rVert_2$ denotes the $\ell_2$-norm of $\boldsymbol{x}$. To obtain the marginal distribution of $\be_\gamma$, we integrate out $\sigma^2$, leading to 
\begin{equation*}\label{eq.post.betagamma}
\be_\gamma|(\Y, \bomega, \re,a,d,\bgamma) \sim t_{2a}\left\{ \hat{\be}_{\gamma} , \frac{d}{a} (\X_\gamma'\X_\gamma)^{-1} \right\}.
\end{equation*}
It is worth noting that the marginal prior of $\bgamma \be$ can decompose into the typical spike-and-slab regression form and that the posterior mean of $\be_\gamma$, $\hat{\be}_{\gamma}$ does not require prior specifications. For random effects $\re$, letting $\bPsi = \left(\mbf{\mathcal{V}}' \mbf{\mathcal{V}} + \sigma^2 \mbf{\mathcal{G}}^{-1} \right)$, with analogous $\bPsi_i = \left(\V_i' \V_i + \sigma^2 \mbf{G}^{-1} \right)$, the conditional distribution is
\begin{equation*}\label{eq.post.re}
\re |(\Y, \sigma^2, \mbf{G}, \bgamma, \be, \bomega) \sim N \left\{ \bPsi^{-1} \mbf{\mathcal{V}}' \left(\Y - \X' \hat \be_\gamma - \V' \bomega \right), \left(\sigma^{-2} \bPsi\right)^{-1} \right \},
\end{equation*}
as is common in Bayesian linear mixed models \citep{Fearn1975, Harville1996, Lange1992, Lindley1972}. For the precision matrix of the random effects, the conditional distribution follows a Wishart distribution,
\begin{equation*}\label{eq.post.g}
\mbf{G}^{-1} |(\Y, \mbf{C}, \rho, \sigma^2, \re, \bgamma, \be, \mbf{\omega})  \sim \mathcal{W} \left\{ \left( \re \re' + \rho \mbf{C}\right)^{-1}, N + \rho \right\}.
\end{equation*}
Finally, for $\bgamma = \left( \gamma_1, \ldots, \gamma_p \right)$ we use empirical Bayes estimators (Section \ref{sec.first.e}) of $\pi$ and $f_{\beta}$ to estimate the posterior expectation $\mbf{p} = \left( p_1, \ldots, p_p \right)$, where $p_k = P\left(\gamma_k = 1| \Y, \pi \right)$. 

\subsection{Parameter expansion and partitioning}\label{sec.pep}

In this section, we present the main aspects of the LMM-PROBE method, which uses a quasi-PX-ECM algorithm to obtain MAP estimates of $\bgamma \be$, $\bomega$, $\re$, $\bG$, and $\sigma^2$. Specifically, we expand the model in (\ref{eq.lmm}) by including latent terms and partitioning the model by the predictor. Let $\X_k$ denote the column of $\X$ corresponding to the $k$th predictor, and let $/k$ indicate a matrix or vector without column or element $k$. Next, consider $\W_k = \X'_{/k} \left( \bgamma_{/k} \be_{/k} \right)$, a term that encompasses the impact of all fixed effect predictors except for $k$, under the sparsity assumption. The partitioned version of model (\ref{eq.lmm}) is then
\begin{eqnarray}\label{eq.part}
\Y|(\W_k, \re) =  \X_k \beta_k + \W_k +  \mbf{\V}' \bomega + \mbf{\mathcal{V}}' \re + \mbf{\epsilon}.
\end{eqnarray}
By expanding $ \X' (\bgamma \be) $ into $ \X_k \beta_k + \W_k $ and partitioning for each $k$, we can conveniently estimate $\beta_k$ while adjusting for the impact of all predictors excluding $k$ (through $\mbf{\V}' \bomega$, $\W_k$) and of all random effects (through $\mbf{\mathcal{V}}' \re$). It is this partitioning (\ref{eq.part}) that allows LMM-PROBE to update each coordinate of the fixed effects in a computationally efficient closed-form - avoiding large matrix operation or intermediate iterative algorithms to numerically solve for parameters.

Note that $\W_k$ and $\re$ are unknown and will be treated as missing data. 
We include expanded parameters on $\W_k$, $\mbf{\V}' \bomega$, and $\mbf{\mathcal{V}}' \re$ to allow for more accurate estimation of the posterior variance of $\beta_k|\gamma_k=1$, since it captures their potential dependence with $\X_k$. These posterior variances -- usually not required or estimated in MAP procedures -- are critical in the empirical Bayes portion of the E-step and facilitate the use of weakly informative priors.
Incorporating expanded parameters $\alpha_k$, $\bomega_k$, and $\tau_k$ in model (\ref{eq.part}), gives
$\Y|(\W_k, \re) =   \X_k\gamma_k \beta_k + \mbf{\V}' \bomega_k + \W_k \alpha_k + \mbf{\mathcal{V}}' \re \tau_k + \mbf{\epsilon}$, 
where to simplify subsequent notation, we use $\bomega_k$ as a coefficient encompassing both $\bomega$ and the expanded parameter. 
Finally, we introduce $\bphi_k = ( \alpha_k, \tau_k)'$ and $\U'_k = [\W_k \quad \mbf{\mathcal{V}}' \re]'$, giving 
\begin{eqnarray}\label{eq.final}
\Y|\U_k =   \X_k\gamma_k \beta_k + \mbf{\V}' \bomega_k + \U'_k \bphi_k + \mbf{\epsilon}.
\end{eqnarray}
Further, we define $\U_{0} = [\W_{0} \quad \mbf{\mathcal{V}}' \re]'$,   where $\W_{0} = \X'(\bgamma\be)$.

We use model (\ref{eq.final}) to perform MAP estimation for $\bgamma\be$, $\mbf{b}$, $\mbf{G}$, and $\sigma^2$. To this end, consider estimating the MAP of $\beta_k|\gamma_k=1$ given $\U_k$ and let $\mbf{\xi}_k=(\beta_k, \bomega_k, \bphi_k)'$ denote the regression parameters in (\ref{eq.final}) given $\gamma_k=1$, with priors $\bomega_k\sim f_\bomega\propto 1$ and $\bphi_k\sim \mbf{f}_\bphi\propto 1$. Standard results can be used to derive the posterior $\mbf{\xi}_k|\Y,\U_k, \sigma^2, \gamma_k = 1  \sim N\left\{\hat{\mbf{\xi}}_k, \sigma^2 (q\Z_k'\Z_k)^{-1} \right\}$, with $\hat{\mbf{\xi}}_k = (\hat{{\beta}}_k, \hat{\bomega}_k, \hat{\bphi}_k)'= (\Z_k'\Z_k)^{-1}\Z_k'\Y$ and $\Z_k = [\X_k \quad \mbf{\V} \quad \U_k]$. Specifically, for parameter of interest $\beta_k$, this gives the posterior $\beta_k|\Y,\U_k,\sigma^2, \gamma_k = 1 \sim N(\hat{\beta}_k, \hat S^2_k )$ where $S^2_k$ denotes the first diagonal element of $\sigma^2 (q\Z_k'\Z_k)^{-1}$. As a result, the MAP of $\beta_k|\gamma_k=1$ is $\hat{\beta}_k$.
  
The parameter expansion and partitioning in (\ref{eq.final}) allows computationally effective MAP estimation via a multi-cycle Expectation-Conditional-Maximization \citep[ECM, ][]{Meng1993} algorithm. A contrast between our approach and PX-EM algorithms is available in \cite{McLain2022}. Section \ref{sec.multicycle} outlines the multi-cycle ECM steps and differentiates LMM-PROBE's computational approach from those of \cite{Schelldorfer2011Aim3}, \cite{Wang2012}, and \cite{Rohart2014}.

\section{PX Multi-cycle ECM 
algorithm}\label{sec.multicycle}

The LMM-PROBE algorithm uses a multi-cycle approach with two M-steps, each followed by an E-step. One complete ECM iteration contains four cycles. We use  $t_{M1}$ and $t_{E1}$  to indicate a quantity estimated in the first M- and E-steps, respectively, with analogous notation ($t_{M2}$, $t_{E2}$) for the second M- and E-steps. The E-steps focus on the expectations of the first two moments of $\U_k$, denoted by $\U_k^{(t_{E1})}$ and $\U_k^{2(t_{E1})}$, respectively, at iteration $t$. As discussed in Section \ref{sec.first.e}, the elements of $\U_k^{(t_{E1})}$ consist of $\W_k^{(t_{E1})}$ and $\re^{(t_{E1})}$, corresponding to the expectations of $\W_k$ and $\re$, respectively. Note that elements containing quantities derived from different cycles are indexed by the latest cycle at which an estimate was updated. Obtaining $\U_k^{2(t_{E1})}$ requires the second moments of $\W_k$ and $\re$ -- denoted by $\W_k^{2(t_{E1})}$ and $\re^{2(t_{E1})}$, respectively --  along with the expectation of $(\W_k\re)$ denoted by $(\W_k \re)^{(t_{E1})}$. The E-steps also update the posterior distributions of $\beta_k$ and $\bgamma$. Finally, the M-steps use these quantities to obtain MAP estimates for $\beta_k$, $\bG$, and $\sigma^2$. Sections \ref{sec.first.e} and \ref{sec.second.m} show how to leverage the properties of the exponential family to compute the moments described in the E-steps, while Sections \ref{sec.first.m} and \ref{sec.second.m} detail the M-steps. 

\subsection{First M-step (M1)}\label{sec.first.m}

In this Section, we describe the first M-step. We introduce subscript $\ell$ to include the $0th$ partition in the M-step, with $\ell \in (0,1,\ldots,p)$, and let $\W=(\W_{0},\W_{1},\ldots,\W_{p})$ and $\bphi = (\phi_0', \phi_1', \ldots, \phi_p')$. For each partition, we maximize the expected complete-data log-posterior distribution of parameters $\bfeta_{\ell} \in \Xi=(\bfeta_0,\bfeta_1,\ldots,\bfeta_P)$, which is denoted by $l(\bfeta_{\ell}|\Y,\U_{\ell}, \mbf{\Gamma}_{\ell})$ where $\mbf{\Gamma}_{\ell}$ are the hyperparameters for $\bfeta_{\ell}$ (Section \ref{sec.pep}). Specifically, we condition on $\Theta = (\btheta_0, \btheta_2, \ldots, \btheta_p)'$, a collection of parameter vectors $\btheta_k = (\beta_k, p_k)$, and maximize the $Q$ function via
\begin{equation}
\hat{\bfeta}_{\ell}^{(t_{M1})} = \mbox{argmax}_{\bfeta_{\ell}} E_{\U_{\ell}}\left\{l(\bfeta_{\ell}|\Y,\U_{\ell}, \mbf{\Gamma}_{\ell})|\Theta_{/\ell}^{(t_{E2}-1)} \right\} \ \mbox{for} \ \ell=0,1,\ldots,p, \nonumber
\end{equation}
where $\btheta_0 = (\bfeta_0, \sigma^2, \bG)$ and $t-1$ represents the iteration prior to $t$. 

The MAP estimates if $\hat{\mbf{\xi}}^{(t_{M1})}_k = \left(\hat{{\beta}}_k^{(t_{M1})}, \hat{\omega}_k^{(t_{M1})}, \hat{\bphi}_k^{(t_{M1})}\right)$ are given by 
\begin{equation*}\label{eq.beta.upd}
\hat{\mbf{\xi}}^{(t_{M1})}_k=\left\{(\Z_{k}'\Z_{k})^{(t_{E2}-1)}\right\}^{-1}\Z^{(t_{E2}-1)'}_{k}\Y,
\end{equation*}
with 
\begin{eqnarray*}\label{xtx}
(\Z_{k}^{\prime} \Z_{k})^{(t_{E2}-1)} = \left[
\begin{array}{ccc}
\X^2_{k} &  \X'_{k} \V & \X_{k}^{\prime} \U^{(t_{E2}-1)}_{k} \\
\V' \X_{k} &  \V'\V & \V' \U^{(t_{E2}-1)}_{k} \\
\U^{(t_{E2}-1)\prime}_{k} \X_{k} & \U^{(t_{E2}-1)\prime}_{k} \V & \U^{2(t_{E2}-1)}_{k}
\end{array}  \right],
\end{eqnarray*}
for $k=1,\ldots,p$. Assuming no additional non-sparse predictors (beyond those included as random effects), the dimensions of $(\Z_{k}^{\prime} \Z_{k})^{(t_{E2}-1)}$ are $2(r+1) \times 2(r+1)$. This is markedly smaller than the $p \times p$ matrices used in other methods \citep{Schelldorfer2011Aim3, Wang2012}. Further, since the updates have a closed form, the M-step does not require additional layers of coordinate descent or parameter tuning \citep{Rohart2014}. 

MAP estimation only gives point estimates of regression parameters. However, the empirical Bayes portion of the E-step (see Section \ref{sec.first.e}) requires estimates of the posterior variance to estimate the hyperparameters. The posterior covariance of $\hat{\mbf{\xi}}^{(t_{M1})}_k$ is estimated via 
\begin{equation*}\label{eq.cov.ksi} 
\sigma^{2(t_{M1})}\left\{(\Z_{k}'\Z_{k})^{(t_{E2}-1)}\right\}^{-1}\left(\Z^{(t_{E2}-1)\prime}_{k}\Z^{(t_{E2}-1)}_{k}\right)\left\{(\Z_{k}'\Z_{k})^{(t_{E2}-1)}\right\}^{-1},
\end{equation*}
where $\hat S^{2(t_{M1})}_k$ denotes its $(1,1)$ element. For elements of $\hat{\bfeta}_0^{(t_{M1})}$, the first M-step results in 
\begin{eqnarray}\label{eq.ksi.zero} 
\hat{\bfeta}_0^{(t_{M1})} = \left\{(\Z_{0}'\Z_{0})^{(t_{E2}-1)}\right\}^{-1}\Z^{(t_{E2}-1)'}_{0}\Y \mbox{, where } \\
(\Z_{0}^{\prime} \Z_{0})^{(t_{E2}-1)} = \left[
\begin{array}{cc}
\V'\V & \V' \U^{(t_{E2}-1)}_{0} \\
\U^{(t_{E2}-1)\prime}_{0} \V & \U^{2(t_{E2}-1)}_{0}
\end{array}  \right] \nonumber.
\end{eqnarray}
The estimates for $\bG^{(t_{M1})}$ and $\sigma^{2(t_{M1})}$ are given by
\begin{eqnarray}\label{eq.G} 
\mbf{G}^{(t_{M1})} &=& \frac{1}{N} \mbf{1}' \re^{2(t_{E2}-1)} \text{, and} \\
\sigma^{2(t_{M1})} 
&=& \frac{1}{M} E\left\{ \mbf{\varepsilon}' \mbf{\varepsilon} | (\Y, \btheta_0^{(t_{E2}-1)} )\right\} \nonumber \\
&=& \frac{1}{M} \mbf{1}' \left[ Tr\left\{ \mbf{\mathcal{V}}' \mbf{\mathcal{V}} \left(\sigma^{-2(t_{M2}-1)} \bPsi^{(t_{M2}-1)}\right)^{-1}\right\} +  \mbf{\varepsilon}^{(t_{M1})'} \mbf{\varepsilon}^{(t_{M1})} \right],
\label{eq.sigma} 
\end{eqnarray}
and $\mbf{\varepsilon}^{(t_{M1})} = (\Y - \V' \hat{\bomega}_0^{(t_{M1})} - \U_{0}^{(t_{E2}-1)\prime} \hat{\bphi}_0^{(t_{M1})} )$.

\subsection{First E-step (E1)}\label{sec.first.e}

In this Section, we describe the four components of the first E-step. First, we update the posterior distribution of $\beta_k$ based on the MAP estimates from the first M-step. Second, we use an empirical Bayes estimator to obtain the posterior expectation $p_k$ of $\gamma_k$. Third, we update the expectation and variance of $\W_{\ell}$. Fourth, we calculate the first two moments of $\re$, the remaining element of $\U_{\ell}$, and obtain $(\W_0 \re)^{(t_{E1})}$. 

To help convergence, we use learning rates $q$ to limit the step size of $\beta^{(t_{E1})}_k$ and $S^{2(t_{E1})}_k$ estimates across iterations, giving
\begin{gather}
\beta^{(t_{E1})}_k = (1-q^{(t_{E1})})\beta^{(t_{E1}-1)}_k  + q^{(t_{E1})}\hat{\beta}^{(t_{M1})}_k, \ \mbox{and} \label{eq.beta} \\ 
S^{2(t_{E1})}_k = \left\{(1-q^{(t_{E1})})(S^{2(t_{E1}-1)}_k)^{-1} + q^{(t_{E1})}(\hat{S}_k^{2(t_{M1})})^{-1}\right\}^{-1} \label{eq.s}.
\end{gather}

The $q^{(t_{E1})}$ values determine the contribution of previous estimates to current estimates, akin to damping or momentum \citep{Min01}. \cite{Henrici1964}, \cite{Jiang2022}\nocite{Min01}, \cite{Varadhan2008}, and \cite{Vehetal20} have discussed nuances and uses of a learning rate, especially in Expectation-Propagation \citep[EP, ][]{MinLaf02}. Our approach uses a value of $q^{(t_{E1})} = \frac{1}{t+1}$.

Second, we estimate $p_k = P\left(\gamma_k = 1| \Y, \pi_0 \right)$. Since $E(\hat{\beta}_k|\gamma_k=0)=0$ and $E(\hat{\beta}_k|\gamma_k=1)\ne 0$, we propose an empirical Bayes estimator based on the `two-group' approach from the multiple testing literature \citep{Efron2001, Efr08, Sto07, Sun2007}. To build the empirical Bayes estimator, we assume test statistics $\mathcal{T}^{(t_{E1})}_k=\beta^{(t_{E1})}_k/S^{(t_{E1})}_k \sim (1-\gamma_k)f_{\mathcal{Z}}(\cdot) + \gamma_k f_1(\cdot)$, where $f_{\mathcal{Z}}(\cdot)$ is a standard normal distribution while $f_1$ is unknown and depends on $f_\beta$. Our proposed estimator also relies on the proportion of null hypotheses, $\pi_0$, so that the probability of a specific test statistic being null is conditional on all observed $\mathcal{T}^{(t_{E1})}$. Combining these elements, the estimator is then
\begin{equation}\label{eq.p}
p^{(t_{E1})}_k = 1-\frac{\hat \pi^{(t_{E1})}_0 f_0\left(\mathcal{T}_k^{(t_{E1})}\right)}{\hat{f}^{(t_{E1})}\left(\mathcal{T}_k^{(t_{E1})}\right)}. 
\end{equation}
In (\ref{eq.p}), we use the \cite{Storey2004} approach to estimate $\pi_0^{(t_{E1})} = \hat{\pi}_0^{(t_{E1})} = \sum_k I({P}^{(t_{E1})}_k\geq \lambda)/\{p \times (1-\lambda)\}$, where ${P}_k^{(t_{E1})}$ is the two-sided p-value for $\mathcal{T}^{(t_{E1})}_{k}$ and $\lambda=0.1$ \citep{Blanchard2009}. 
We estimate $f^{(t_{E1})}$ with Gaussian kernel density estimation on ${\mbf{\mathcal{T}}}^{(t_{E1})} = (\mathcal{T}^{(t_{E1})}_1,\ldots,\mathcal{T}^{(t_{E1})}_p)$. While the empirical Bayes estimator of $p_k$ does not assign distributional assumptions on the priors for $\bgamma$ and $\be$, it assumes $f_1$ is nonnegative and non-increasing monotonic from zero. 

Next, using components $\be$ and $\mbf{p}$, we compute the first and second moments of $\W_{\ell}$. To ease computation, we perform calculations at the cluster $i$ level: 
\begin{eqnarray}\label{eq.EW}
E(W_{i 0}|\Theta^{(t_{E1})}_{/0}) = W^{(t_{E1})}_{i 0} = E\{ \X_{i} (\bgamma\be) |\Theta^{(t_{E1})}_{/0}\} = \X_{i} (\be^{(t_{E1})} \mbf{p}^{(t_{E1})} ).
\end{eqnarray}
The variance of $\W_{\ell}$ is needed for the second moment
\begin{gather}\label{eq.VarW}
Var(W_{i0}|\Theta^{(t_{E1})}_{/0}) = \X_{i}^2 \left\{\be^{(t_{E1})2}  \mbf{p}^{(t_{E1})}(1- \mbf{p}^{(t_{E1})}) \right\} \\
E(W_{i 0}^2|\Theta^{(t_{E1})}) = W^{2(t_{E1})}_{i 0} = Var(W_{i0}|\Theta^{(t_{E1})}_{/0}) + \left(W^{(t_{E1})}_{i 0}\right)^2. \nonumber
\end{gather}
From $W^{(t_{E1})}_{i 0}$ and $W^{2(t_{E1})}_{i 0}$, we get $\W_{\ell}^{(t_{E1})}$ and  $\W_{\ell}^{2(t_{E1})}$, which are used in both M-steps ($M1$ and $M2$) as well as in the moments of $\U_{\ell}$ (E-steps $E1$ and $E2$). 

We end the first E-step by completing the moment calculations for the remaining element ($\re$) of $\U_{\ell}$ as well as $E(\W_{\ell} \re |\Theta^{(t_{E1})}_{/0})$. $E(\U_{i 0}|\Theta^{(t_{E1})}_{/0})$ is 
\begin{gather}\label{eq.eD}
 \U^{(t_{E1})}_{i 0} = 
    \begin{bmatrix}
   W_{i 0}^{(t_{E1})}  \\
    E(\V_i' \re_i|\Theta^{(t_{E1})}_{/0}) 
    \end{bmatrix}  
    = \begin{bmatrix}
    \X_{i} (\be^{(t_{E1})} \mbf{p}^{(t_{E1})} ) \\
    \V_i' \re_i^{(t_{E1})}
    \end{bmatrix},
\end{gather}
where $\re^{(t_{E1})}$ is estimated via 
\begin{equation}\label{eq.EL}
\re^{(t_{E1})} = \bPsi^{-1(t_{M1})} \mbf{\mathcal{V}}' (\Y-  \V' \hat{\bomega}_0^{(t_{M1})} - \W_{0}^{(t_{E1})'} \hat{\alpha}_0^{(t_{M1})}  ).
\end{equation}
Further, $\re^{2(t_{E1})}$ can be obtained via $Var(\re|\W_{0}^{(t_{E1})}, \btheta_0^{(t_{M1})}) + E(\re|\W_{0}^{(t_{E1})}, \btheta_0^{(t_{M1})})^2,$ giving
\begin{eqnarray}\label{eq.EL2}
\re^{2(t_{E1})} = ( \sigma^{-2(t_{M1})} \bPsi^{(t_{M1})})^{-1} + \re^{(t_{E1})} \re^{(t_{E1})\prime},
\end{eqnarray} 
where $\bPsi^{(t_{E1})} = \left(\mbf{\mathcal{V}}' \mbf{\mathcal{V}} + \sigma^{2(t_{M1})} \mbf{\mathcal{G}}^{-1(t_{M1})} \right)$. Finally, $\U_{i 0}^{2(t_{E1})}$ is estimated via \\ $Var\left(\U_{i 0}|\Theta^{(t_{E1})}_{/0} \right) + \U_{i 0}^{(t_{E1})\prime} \U_{i 0}^{(t_{E1})} $, with 
\begin{eqnarray}\label{eq.varD}
Var \left( \U_{i 0}|\Theta^{(t_{E1})}_{/0} \right) = 
    \begin{bmatrix}
     \X_{i}^2 \left\{\be^{(t_{E1})2}  \mbf{p}^{(t_{E1})}(1- \mbf{p}^{(t_{E1})}) \right\}  & Cov(W_{i0},\V_i' \re_i|\Theta^{(t_{E1})}_{/0}) \\
      & \V_i'( \sigma^{-2(t_{M1})} \bPsi_i^{(t_{M1})} )^{-1} \V_i 
    \end{bmatrix} \label{eq.covWL}
\end{eqnarray}
where
$Cov(W_{i0}, \V_i' \re_i|\Theta^{(t_{E1})}_{/0}) =  -\V_i \bPsi_i^{-1(t_{M1})} \V_i^{\prime} Var(W_{i0}|\Theta^{(t_{E1})}_{/0}).$

\subsection{Second M- and E-steps (M2, E2)}\label{sec.second.m}

The second M-step updates the estimates of the $0th$ partition of the first M-step with $\U^{(t_{E1})}$ and $\re^{2(t_{E1})}$ via (\ref{eq.ksi.zero})-(\ref{eq.sigma}). This results in $\hat{\bfeta}_0^{(t_{M2})} = (\hat{\bomega}_0^{(t_{M2})}, \hat{\bphi}_0^{(t_{M2})})$,  $\bG^{(t_{M2})}$, and $\sigma^{2(t_{M2})}$.  In the second E-step (E2), for $\W_{0}^{(t_{E2})}$ and $\W_{0}^{2(t_{E2})}$, we simply write $\W_{0}^{(t_{E2})} = \W_{0}^{(t_{E1})}$ as well as $\W_{0}^{2(t_{E2})} = \W_{0}^{2(t_{E1})}$ (these elements do not change). The moments $\re^{(t_{E2})}$ and $\re^{2(t_{E2})}$ are updated conditional on $ \btheta_0^{(t_{M2})} = (\hat{\bfeta}_0^{(t_{M2})}, \bG^{(t_{M2})}, \sigma^{2(t_{M2})})$ using (\ref{eq.EL})--(\ref{eq.EL2}). The updated $\re^{(t_{E2})}$ along with  $\W_{\ell}^{(t_{E2})}$ give $\U^{(t_{E2})}$. The expectation $(W_{i0} \re_i )^{(t_{E2})}$ is also updated in this cycle using $E(W_{i0} \V_i' \re_i |\mbf{\btheta}^{(t_{E2})}_{/0}) =  \V_i \bPsi_i^{-1(t_{M2})}  \V_i' \left( \Y_i W^{(t_{E2})}_{i0}- W^{2(t_{E2})}_{i0} \right)$. Finally, the updated $\re^{2(t_{E2})}$, $(W_{i0} \re_i )^{(t_{E2})}$, along with  $\W_{0}^{2(t_{E2})}$ give an updated $\U^{2(t_{E2})}$.

\subsection{Algorithm, estimates and predictions}\label{sec.algo}

Algorithm \ref{algo.1} shows the steps from Sections \ref{sec.first.m}, \ref{sec.first.e}, and \ref{sec.second.m} in sequence. Because of the independence between predictors $k$ in $\W_k$, computations are relatively inexpensive, making the algorithm efficient. To initiate the algorithm, we use $\be^{(0_{E1})} = \mbf{0}$, $\re^{(0_{E2})} = \mbf{0}$, $\mbf{p}^{(0_{E1})}=\mbf{0}$, $\sigma^{2(0_{M1})}$ as the sample variance of $\Y$, and $\bG^{(0_{M1})} = \mbf{I}_r$, which results in $\U_{i\ell}^{(0_{E2})}=\U_{i\ell}^{(0_{E2})}=\mbf{0}$. Initiating LMM-PROBE is flexible and other values can be used \citep[see][for other options]{McLain2022}. Convergence is defined by changes in $\W_{0}^{(t_{E2})}$ between successive iterations, with small changes indicating convergence. Formally, let $C_{ij}^{(t_{E2})}= (W^{(t_{E2})}_{ij 0} - W_{ij 0}^{(t_{E2}-1)})^2/ Var(W_{ij0}|\Theta^{(t_{E2-1})}_{/0})$ quantify the change in $W_{ij 0}$ between steps $t$ and $t-1$. Our convergence criterion is $CC^{(t_{E2})} = \log(M)\max_{ij}\left(C_{ij}^{(t_{E2})}\right)$, where index $j$ represents observations. The use of $\log(M)$ adjusts for the effect of sample size on the maximum of Chi-squared random variables \citep{Embetal13}. We stop the algorithm after the first E2 cycle when $CC^{(t_{E2})} < \chi^2_{1,0.1}$, where $1$ and $0.1$ represent the degrees of freedom and the quantile of a Chi-squared distribution, respectively. 

\begin{algorithm}[t]
\caption{LMM-PROBE algorithm sequence}\label{algo.1}
\begin{algorithmic}[]
 \State  Initiate $\U^{(0_{E2})}$ and $\U^{2(0_{E2})}$
    \While{$CC^{(t_{E2})} \geq \chi^2_{1,0.1} $ and $\max(\mbf{p}^{(t_{E1})})>0$} 
   \State  \textbf{M-step M1}
        \State \hspace{0.6 cm} (a) Use $\U^{(t_{E2}-1)}$ and $\U^{2(t_{E2}-1)}$ to estimate $\bfeta^{(t)}_{\ell}$ for $\ell=0,1,\ldots,p$.
        \State \hspace{0.6 cm} (b) Calculate $\bG^{(t_{M1})}$ via (\ref{eq.G}).
        \State \hspace{0.6 cm} (c) Calculate $\sigma^{2(t_{M1})}$ and use to estimate $\hat{{S}}_k^{(t_{M1})2}$ via (\ref{eq.sigma}) for all $k$.
\State \textbf{E-step E1}
        \State\hspace{0.6 cm} (a) Calculate $\beta^{(t_{E1})}_k$ and $S^{2(t_{E1})}_k$ using (\ref{eq.beta}--\ref{eq.s}) for all $k$.
    \State\hspace{0.6 cm} (b) Estimate $\hat{f}^{(t_{E1})}$ and $\hat{\pi}^{(t_{E1})}_{0}$ and use them to calculate $\mbf{p}^{(t_{E1})}$ via (\ref{eq.p}).
    \State\hspace{0.6 cm} (c) Calculate $\mbf{U}^{(t_{E1})}$, $\mbf{U}^{2(t_{E1})}$ using (\ref{eq.EW})--(\ref{eq.covWL}).
\State  \textbf{M-step M2}
        \State \hspace{0.6 cm} (a) Use $\U_0^{(t_{E1})}$ and $\U_0^{2(t_{E1})}$ to estimate $\bfeta^{(t_{M2})}_{0}$.
        \State \hspace{0.6 cm} (b) Calculate $\bG^{(t_{M2})}$ via (\ref{eq.G}).
        \State \hspace{0.6 cm} (c) Calculate $\sigma^{2(t_{M2})}$ via (\ref{eq.sigma}).
\State \textbf{E-step E2}
    \State\hspace{0.6 cm} (a) Set $\mbf{W}^{(t_{E2})}$ and $\mbf{W}^{2(t_{E2})}$ equal to $\mbf{W}^{(t_{E1})}$ and $\mbf{W}^{2(t_{E1})}$, respectively. 
    \State\hspace{0.6 cm} (b) Calculate $\mbf{U}^{(t_{E1})}$, $\mbf{U}^{2(t_{E1})}$ using (\ref{eq.EW})--(\ref{eq.covWL}).
    \EndWhile  \label{roy's loop}
\end{algorithmic}
\end{algorithm}

After convergence, LMM-PROBE results in MAP estimates $\tilde \be$, $\tilde{\mbf{S}}^2$, $\tilde{\mbf{p}}$, $\tilde{\bphi}_0$, $\tilde{\bomega}_0$, $\tilde \sigma^2$, $\tilde \bG$, $\tilde \U_{0}$, $\tilde \U_{0}^2$, as well as the MAP of $\bgamma_k \beta_k$, $\bar \be =  \tilde{\alpha_0} \tilde \be$, if we assume $\bgamma = \mbf{1}$, or  $\bar \be =  \tilde{\alpha}_0\tilde{\mbf{p}} \tilde \be$ if we combine the expected $\bgamma \be$ with the MAP of $\alpha_0$. The full posterior predictive distributions of $W_{l 0}$, $\re_{l}$ and $\Y_{l}$ (for test data $\X_l$ and $\V_l$) are not available with our estimation procedure. As a result, alternative predictions are obtained via $ \tilde{Y}_{l} = \V_{l}'\tilde{\bomega}_0 +  \tilde{\bphi}_0 \tilde{\U}_{l 0}$. The elements of $\tilde{\U}_{l 0}$ (i.e., $\tilde{W}_{l 0}$, $\tilde{\re}_{l}$) are obtained by plugging in new data $\X_{l}$, $\V_{l}$ and other MAP estimates into equations (\ref{eq.eD})--(\ref{eq.varD}).

\subsection{Computational Complexity}\label{sec.comp.anal}

As demonstrated in Section \ref{sec.intro}, LMM-LASSO \citep{Schelldorfer2011Aim3}, LASSO+ \citep[implemented in the \texttt{MMS} package,][]{Rohart2014}, and PGEE \citep{Wang2012} resulted in long computation times for large $p$. Below, we investigate the computation times for more settings and find that estimation with LMM-PROBE is markedly quicker. This is not unexpected considering the computational complexity of the algorithms. For example, each iteration of the PGEE approach requires the inverse of a $p\times p$ matrix and LMM-LASSO requires calculating a $p\times p$ matrix of second derivatives of the objective function. These methods have computational complexity that is lower-bounded by $\Omega(p^3)$ and $\Omega(p^2)$, respectively. As a result, each is computationally expensive and less scalable for large $p$.

LASSO+, which scales better than LMM-LASSO and PGEE, has a similar computational complexity to LMM-PROBE for the \textit{random effect} portion of the algorithm. However, the updates of the fixed effects for LASSO+ require fitting a LASSO model -- which has computational complexity $O\{K_{L}M(p+k)\}$ where $k$ is the number of non-zero coefficients and $K_{L}$ is the number of iterations of the LASSO \citep{Hastie2015}. With no additional non-sparse predictors, the LMM-PROBE M1-step requires $p+1$ linear regression models, each with complexity $O\{(2r+2)^2M+(2r+2)^3\}$, while updates of the E-step require $O(pM)$ complexity. A full update of the fixed effects requires $O[p\{M+(2r+2)^2M+(2r+2)^3\}]$ complexity. As a result, the fixed effect computational complexity of LMM-PROBE and LASSO+ grows linearly with $p$ and $M$, but each LASSO+ iteration requires $K_{L}$ LASSO iterations and is repeated for multiple $\lambda$ penalty values. Further, $K_{L}$, $k$, and the number of iterations required by LASSO+ tend to increase with $p$ hurting the scalability of the procedure.

\subsection{PX-ECM vs EM}\label{sec.conv}

\textcolor{black}{Our use of the PX multi-cycle ECM approach for LMM-PROBE is motivated by practical considerations. First, a standard EM for the Bayesian LMM in (\ref{eq.lmm}) is not always identifiable without imposing additional prior restrictions and requires the inversion of a $p \times p$ matrix in each iteration. Therefore, an ECM approach is beneficial. 
This is formalized in Proposition \ref{prop1}, which proves that in the context of LMM-PROBE, the parameters are not always estimable for a standard EM algorithm, which involves computationally expensive calculations. Second, adding parameter expansion (i.e., PX) can speed convergence \citep{Liuetal98}, and suits our interest in estimating the effect of a predictor $k$ while accounting for (versus correcting for) the impact of all other predictors and random effects. See Remarks 1 and 2 in Section \ref{sec.pxecm.em} of the Supplementary Materials for related results and discussion. As a result, we formulated LMM-PROBE as a PX multi-cycle ECM. Proposition \ref{prop2} asserts that the M-steps of our PX-ECM framework can always be solved.}

\begin{proposition} \label{prop1}
\textcolor{black}{Under a standard EM algorithm, the parameters in the M-step for LMM-PROBE are not always estimable. That is, the maximizer of $Q_{\text{EM}}(\boldsymbol{\eta} \mid \Theta^{(t-1)})$, where $\boldsymbol{\eta} = (\be \ \bomega \ \btau)'$,
is not always unique.} 
\end{proposition}

\begin{proposition} \label{prop2}
\textcolor{black}{Let
$Q_{\text{CM}}^{M1}(\boldsymbol{\eta} \mid \Theta^{(t-1)})$ and $Q_{\text{CM}}^{M2}(\boldsymbol{\eta} \mid {\Theta}^{(t-1)})$ denote the two M-step quantities in the PX-ECM algorithm for LMM-PROBE.
Assuming no perfect collinearity between $\X_k$ and $\V$ for any $k$, the maximizers of $Q_{\text{CM}}^{M1}(\boldsymbol{\eta} \mid {\Theta}^{(t-1)})$ and $Q_{\text{CM}}^{M2}(\boldsymbol{\eta} \mid {\Theta}^{(t-1)})$ always exist and are unique.} 
\end{proposition}

\textcolor{black}{The proofs of Propositions \ref{prop1} and \ref{prop2} can be found in Section \ref{sec.pxecm.em} of the Supplementary Material. Taken together, these propositions motivate the use of the PX-ECM algorithm. In Section \ref{sec.sm.conv} of the Supplementary Material, we provide further convergence assessments of PX-ECM.}

\section{Simulation Studies}\label{sec.sim}

We performed numerous simulation studies to evaluate the performance of LMM-PROBE with regard to the estimation of fixed effects, variance components, and predictions of future values. We defined the outcome as $\Y = \X (\bgamma \mbf{\beta})  + \V' \mbf{\omega}^\ast  + \mbf{\mathcal{V}}' \re + \mbf{\epsilon}$, where $\mbf{\epsilon} \sim N(0, \sigma^2 \I)$, and $\re \sim N(0, \bG)$. To generate data with dependence within $\mbf{\gamma}$ and $\X$ we used a Gaussian random field (GRF) where the covariance is a squared exponential function \citep{Schlather2015}. $\bgamma$ represents the elements of the GRF greater than a threshold such that $\pi = |\bgamma|/p$, where $\pi$ is a simulation setting. 
For each cluster $i$, we generated $n_i$ observations, of which the first half were used in the training set while the remaining half were used to estimate prediction error. 

Simulation settings varied across different values of $p$, $N$, $M$, $\pi$, $\bG$, $\sigma^2$, $\mbf{\beta}$ ($\bomega$), and $r$. Specifically, we used $p \in (15^2, 25^2, 75^2)$, $\pi \in (0.05, 0.1)$, $\mbf{\beta}, \bomega \in (0.50, 0.75)$, and $r \in (1, 2)$ random effects where $\V_{i} = \mbf{1}'$ for $r = 1$ and $\V_{i} = [\mbf{1}' \ (1, \ldots, n_i)']$ for $r=2$. The remaining values $(N,n_i,\sigma^2)$ were based on $p$. For $p = 15^2$, $N = 50$, $n_i = 6$, and $\sigma^2 \in (\sigma^2_1, \sigma^2_2) = (10, 15)$, for $p = 25^2$, $N = 100$, $n_i = 6$, and $\sigma^2 \in (10, 15)$, and for $p = 75^2$, $N = 250$, $n_i = 8$, and $\sigma^2 \in (100, 150)$. Additionally, two values of the random effect covariance matrix $\bG$ were considered for each setting. These values depended on $p$ and $r$,
\begin{equation}\label{cases}
\bG \in  (G_1,G_2) =\begin{dcases}
    (5, 10)  & \text{ if } p=15^2, r=1 \text{ or } p=25^2, r=1 \\
    \left(\renewcommand\arraystretch{0.75}
    \begin{bmatrix}
    4 & 0 \\ 
    0 & 2.5
    \end{bmatrix} , 
    \begin{bmatrix}
    6 & 1 \\
    1 & 3.5
    \end{bmatrix}\right) & \text{ if } p=15^2, r=2 \text{ or } p=25^2, r=2 \\
    (50, 100)  & \text{ if } p=75^2, r=1 \\
     \left( \renewcommand\arraystretch{0.75}
    \begin{bmatrix}
    40 & 0 \\ 
    0 & 25
    \end{bmatrix} , 
    \begin{bmatrix}
    60 & 10 \\
    10 & 35
    \end{bmatrix}\right) & \text{ if } p=75^2, r=2. \nonumber
\end{dcases} 
\end{equation}
Note that $||G_1||<||G_2||$ for any norm $||\cdot||$.

We compared LMM-PROBE to PROBE \citep{McLain2022}, two types of linear mixed modeling with LASSO penalty \citep[LMM-LASSO, LASSO+, ][respectively]{Schelldorfer2011Aim3, Rohart2014}, LASSO \citep{Tib96}, and penalized general estimating equations \citep[PGEE, ][]{Wang2012}. LMM-PROBE, PROBE, and LASSO simulations included 500 iterations for all simulation settings. LASSO+, LMM-LASSO, and PGEE simulations included 250 iterations due to computation time requirements. Further, LASSO+ only covered settings where $p \in (15^2, 25^2)$ while LMM-LASSO and PGEE covered settings where $p = 15^2$, as these methods were not feasible with larger $p$. All LMM-PROBE simulations used $\chi^2_{1,0.1}$ to evaluate convergence and $q^{(t)} = \frac{1}{t+1}$ as the learning rate. For LMM-LASSO and LASSO+, we used five-fold CV to select the tuning parameter that minimized the Bayesian Information Criterion. For LASSO and PGEE, we used the default CV in the $\verb|glmnet|$ \citep{glmnet} and $\verb|PGEE|$ \citep{pgee} packages. For parameters that did not require tuning, we used package defaults. Due to the difference in the results for PGEE compared to other methods, figures including PGEE results are provided in Supplementary Materials only. 

To examine the combined fixed and random effect estimates, we calculated Mean Squared Prediction Error (MSPE) of test data. Specifically, we calculate $\tilde{\Y}_l$ based on $\X_l$ and $\V_l$ for $N$ \textit{test} subjects not used in the fitting of the model. We then averaged the squared prediction errors, $(\Y_l - \tilde{\Y}_l)^2$. Figure \ref{fig.mse.simul} shows MSPEs for LMM-PROBE compared to other methods, for various $(p, \sigma^2,\bG)$ settings with $r = 2$ and $\pi = 0.1$. LMM-PROBE had the lowest MSPE across all simulation settings. Further, PROBE and LASSO had lower MSPEs than LMM-LASSO and LASSO+ across all settings. Comparing the results for $\bG = G_1$ to $G_2$, it appears that the MSPE for LMM-PROBE is more robust to increasing the variance of random effects, whereas the MSPE for comparison methods increased more markedly. An increase in residual variance ($\sigma^2$ = $\sigma^2_1$ to $\sigma^2_2$) resulted in a slight increase in MSPE. MSPE results remained largely the same when examining settings where $r = 1$ and $\pi = 0.1$, shown in Supplementary Materials Figure \ref{fig.mspe.1re}. In all settings except one, LMM-PROBE displayed the lowest MSPE, followed by PROBE and LASSO, with LMM-LASSO and LASSO+ displaying the highest MSPEs (Figure \ref{fig.mspe.1re}). This setting where LMM-PROBE displayed a higher MSPE occurred in the scenario with the lowest Intracass Correlation Coefficient (ICC) among all tested (ICC: $50/(50+150)$, for factors $\bG =50 $,  $\sigma^2 = 150$, $\mbf{\beta} = 0.5$). 

\begin{figure}[ht]
\centering
\includegraphics[width=7in]{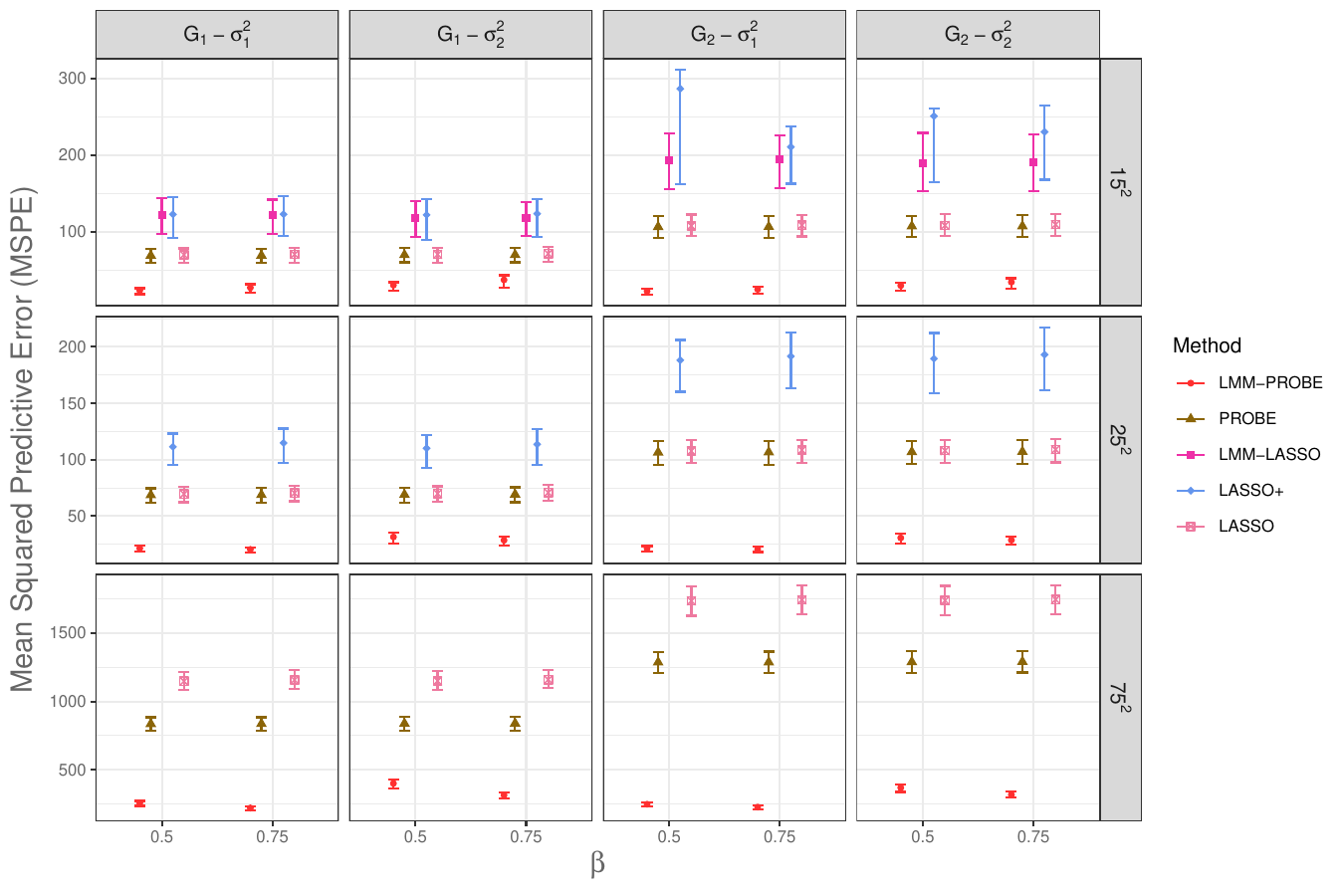} \\
\caption{Mean Squared Predictive Errors (MSPE) for LMM-PROBE and four comparison methods across various simulation settings, including $p$, $\sigma^2$, $\bG$, and $\mbf{\beta}$ values, when $r = 2$ and $\pi = 0.1$. The MSPEs are based on both fixed and random effects. Vertical lines display the interquartile range of the MSPEs. Comparison methods LMM-LASSO and LASSO+ are methods for linear mixed models, while LASSO and PROBE are methods for linear models.  \label{fig.mse.simul}}
\end{figure} 

In the above comparisons, LMM-PROBE has an inherent advantage of using random effect predictions over non-LMM approaches. As a result, we also examined the MSEs focusing on the fixed effects only ($\X_{i} \bgamma \mbf{\beta} - \X_{i} \tilde{\alpha_0} \tilde{\mbf{p}}\tilde{\be}$), displayed in Figure \ref{fig.mse.simul.fix} for $r = 2$ and $\pi = 0.1$. The results followed the same trends as in Figure \ref{fig.mse.simul}, with LMM-PROBE showing the lowest MSEs across simulation settings. The Supplementary Materials Section \ref{sec.sm.sim} includes Median Absolute Deviations (MADs) in Figure \ref{fig.mad.1re}, where trends were highly similar to the MSE trends. 
\begin{figure}[ht]
\centering
\includegraphics[width=7in]{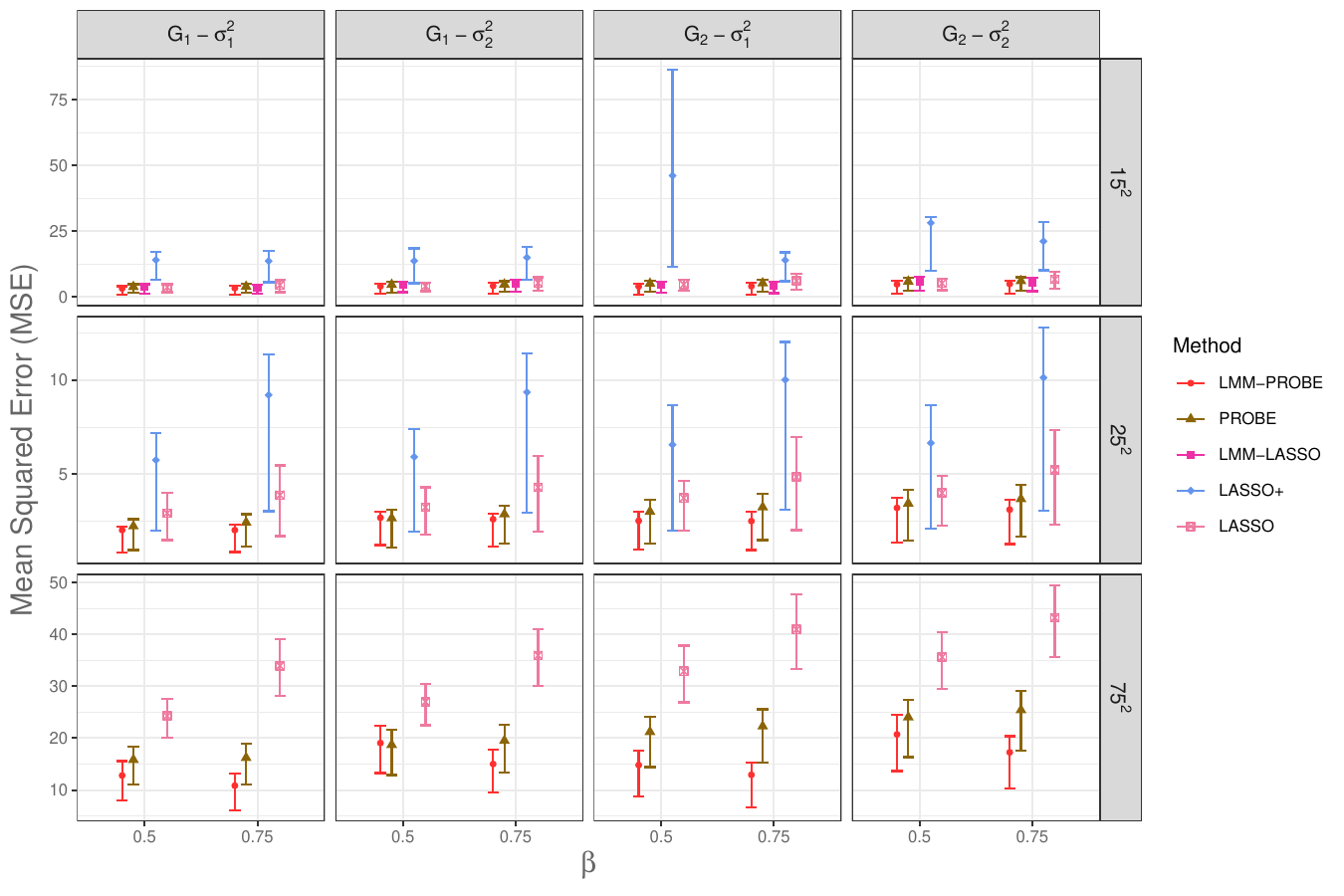} \\
\caption{Mean Squared Errors (MSE) for LMM-PROBE and four comparison methods across various simulation settings, including $p$, $\sigma^2$, $\bG$, and $\mbf{\beta}$ values, when $r = 2$ and $\pi = 0.1$. The MSEs are based on fixed effects only. Vertical lines display the interquartile range of the MSEs. Comparison methods LMM-LASSO and LASSO+ are methods for linear mixed models, while LASSO and PROBE are methods for linear models.  \label{fig.mse.simul.fix}}
\end{figure} 
Further, we examined the Mean Squared Error (MSE) of the total marginal variance $\mbf{\mathcal{V}}' \mbf{\mathcal{G}} \mbf{\mathcal{V}} + \sigma^2 \I$ for the methods that estimate $\bG$ and $\sigma^2$ (LMM-PROBE, LMM-LASSO, LASSO+). Figure \ref{fig.mse.totvar} shows that when $p = 15^2$, LMM-PROBE and LMM-LASSO had total variance estimates with similar error, which was lower than that of LASSO+, especially in simulation scenarios where $r = 2$. For $p = 25^2$, the MSE was similar between LMM-PROBE and LASSO+, with a slightly higher MSE for LMM-PROBE for some $r = 2$ settings. 
\begin{figure}[ht]
\centering
\includegraphics[width=7in]{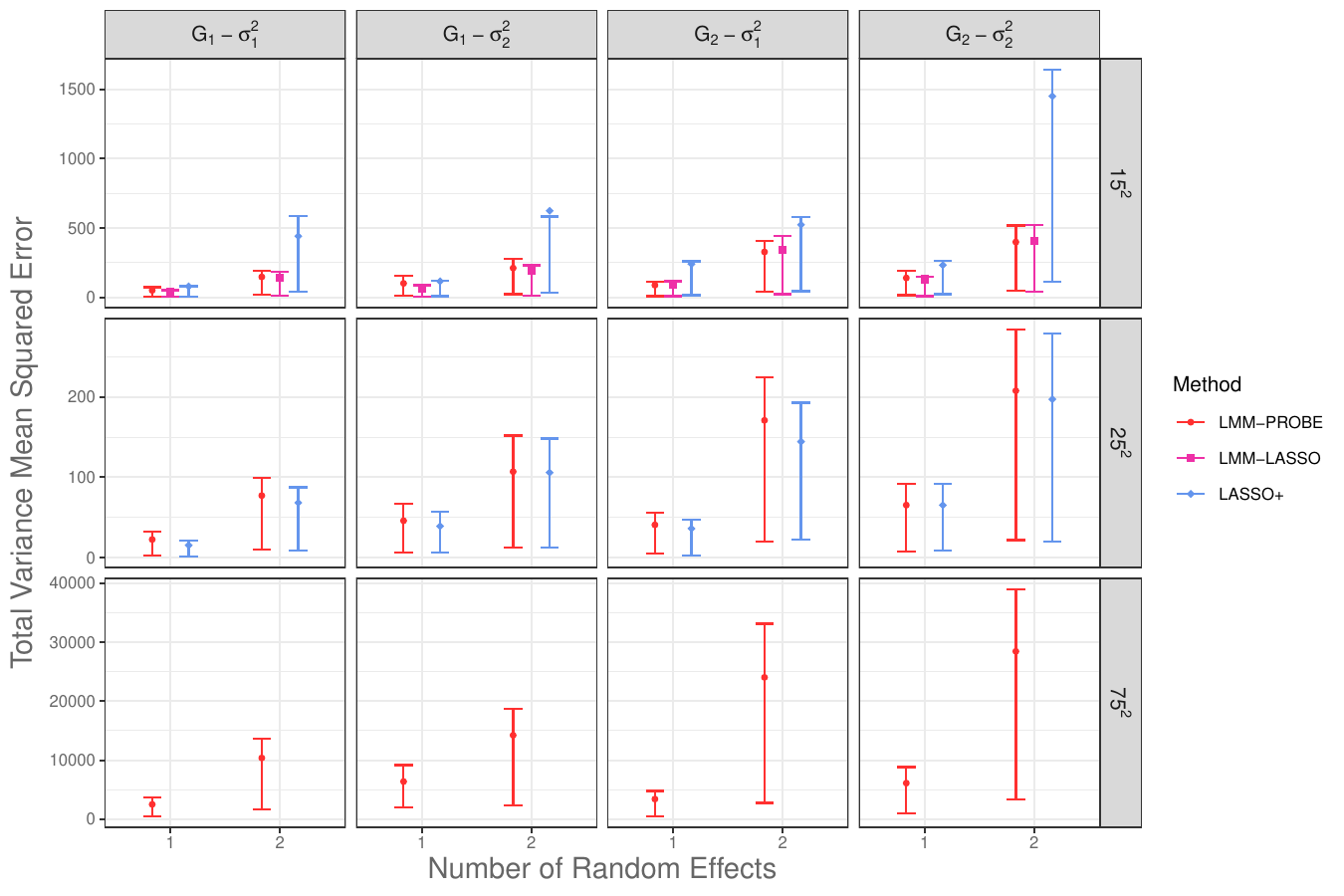} \\
\caption{Mean Squared Error (MSE) of the total model variance ($\mbf{\mathcal{V}}' \mbf{\mathcal{G}} \mbf{\mathcal{V}} + \sigma^2 \I$) for LMM-PROBE, LMM-LASSO, and LASSO+ across various simulation settings, including $p$, $\sigma^2$, $\bG$, and $r$ values, when $\mbf{\beta} = 0.75$ and $\pi = 0.1$. Vertical lines display the interquartile range of the MSEs. \label{fig.mse.totvar}}
\end{figure} 

\textcolor{black}{We also examined the variable selection performance of LMM-PROBE and other methods. Results including sensitivity, specificity, and the Matthews Correlation Coefficient \citep[MCC, ][]{Matthews1975} are displayed in Figure \ref{fig.varselect}. Note that variable selection results were similar across the different levels of $\sigma^2$ and $\bG$; therefore, Figure \ref{fig.varselect} shows results when $\sigma^2 = \sigma^2_1$ and $\bG = \bG_1$. For LMM-PROBE and PROBE, a predictor was selected if $\tilde{p}_k > 0.5$. For all other methods, a predictor was selected if its estimated $\tilde{\beta}_k \ne 0$. PGEE resulted in a sensitivity of 1 and a specificity of 0 for each setting. For brevity, this method is omitted from the figures. Figure \ref{fig.varselect} shows that when $p = 15^2$ and the proportion of signals was lower ($\pi = 0.05$), LMM-PROBE had a similar sensitivity as PROBE, and outperformed other methods. For $\pi = 0.1$, the sensitivity of LMM-PROBE did not match that of PROBE but exceeded all other methods, including methods for LMMs. As the number of random effects went from $r = 1$ to $r = 2$, LMM-PROBE had higher sensitivity than PROBE. LMM-PROBE had higher specificity in all simulation settings, especially in settings where the proportion of signals was higher ($\pi = 0.1$). This led to LMM-PROBE having a higher MCC in all settings but two. Additional variable selection results are available in Supplementary Materials Figures \ref{fig.varselect625} ($p = 25^2$) and \ref{fig.varselect5625} ($p = 75^2$) in Section \ref{sec.sm.sim}. }

Finally, Figure \ref{fig.time} shows the average computation time in minutes per simulation iteration. Overall, the running time for LMM-PROBE was markedly faster than the other approaches that also modeled random effects. Specifically, LMM-LASSO (when $p = 15^2$) and LASSO+ (when $p = 25^2$) required extensive computation time, averaging nearly one hour per iteration. For LASSO+, computation time was notably impacted by the number of random effects $r$. PGEE required over five minutes per iteration for $p = 15^2$, but did not scale well and could not be used for larger $p$ settings. As expected, PROBE and LASSO were the most time-effective. LMM-PROBE required seven minutes on average when $p = 75^2$, much less than some comparison methods when $p$ was lower $(15^2, 25^2)$. All simulations were performed on an Intel Xeon 8358 Platinum processor with 2.6GHz CPU and 128 GB memory.

\begin{figure}[ht]
\centering
\includegraphics[width=7in]{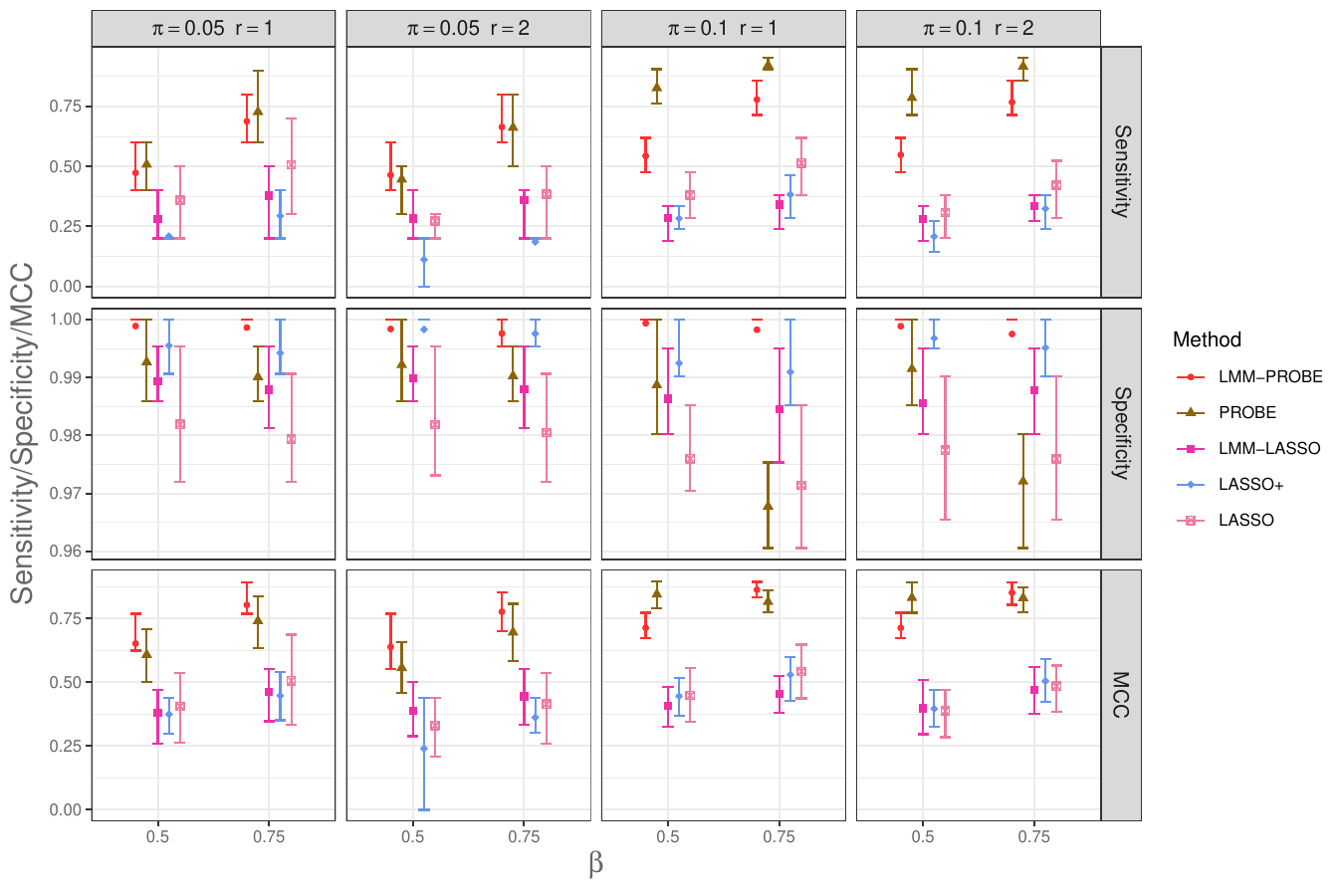} \\
\caption{Sensitivity, Specificity, and the Matthews Correlation Coefficient (MCC) for LMM-PROBE and five comparison methods across various simulation settings, including $\pi$, $r$, and $\be$ values, when $\sigma^2 = \sigma^2_1$, $\bG = \bG_1$, and $p = 15^2$. Vertical lines display the interquartile range of the sensitivity, specificity, and MCC. Comparison methods LMM-LASSO and LASSO+ are methods for linear mixed models, while LASSO and PROBE are methods for linear models. \label{fig.varselect}}
\end{figure} 

\begin{figure}[ht]
\centering
\includegraphics[width=7in]{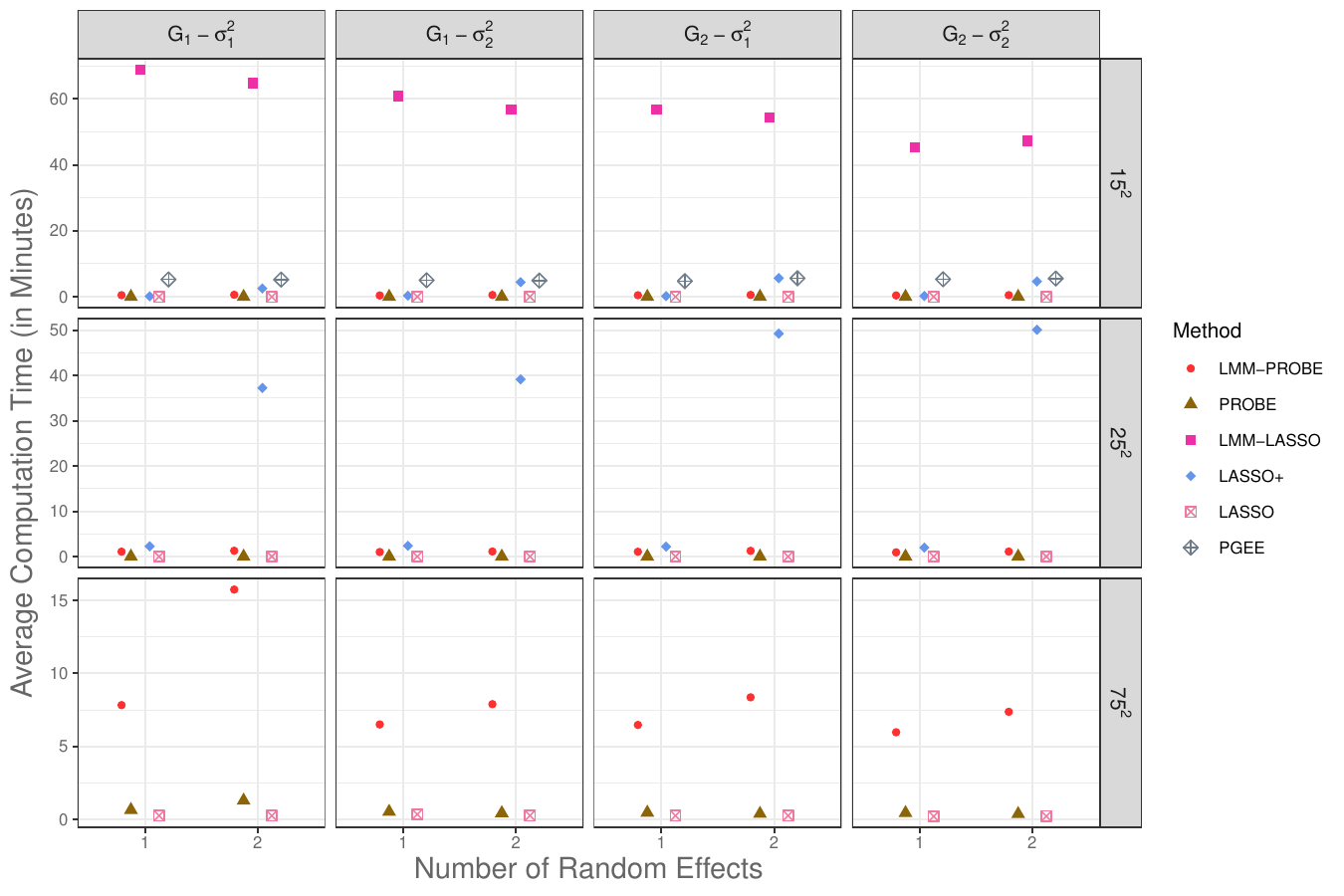} \\
\caption{Average computation time in minutes per simulation iteration for LMM-PROBE and five comparison methods across various simulation settings, including $p$, $\sigma^2$, $\bG$, and $r$ values, averaged over $\mbf{\beta}$ and $\pi$. \label{fig.time}}
\end{figure}

\section{Real Data Application}\label{sec.data}

We showcase the performance and characteristics of the LMM-PROBE method on a study of systemic lupus erythematosus (SLE) with a cohort of 158 pediatric patients receiving treatment at rheumatology clinics at Texas Scottish Rite Hospital for Children and Children’s Medical Center Dallas \citep{Banchereau2016}. SLE is a systemic autoimmune disease with recurring flares that can cause damage to organs over time. In recent years, studies have focused on ways to identify and diagnose SLE. The study conducted by \cite{Banchereau2016} was key as it introduced high-dimensional longitudinal measurements of SLE biomarkers. A critical finding in \cite{Banchereau2016} is the overexpression of the IFI6 biomarker (`Interferon alpha-inducible protein 6') in SLE patients. Until transcription of IFI6 becomes available as a standard diagnostic tool of SLE in clinical practice, practitioners rely on factors related to overexpression of IFI6 to guide their diagnosis of SLE.
The \cite{Banchereau2016} data is available from the Gene Expression Omnibus database hosted by the National Center for Biotechnology Information, using accession number GSE65391. 

Our analysis examines genetic and clinical predictors of IFI6 expression levels. We use 15386 predictors representing gene expression data of blood sample components (e.g., cells related to disease sites or the lymphatic system). The dataset also includes 38 clinical predictors, representing complete blood count data, demographic information, and symptomatology. We retain patients with complete genetic and clinical observations and those patients with observations beyond the baseline visit ($N=125$). Our final dataset included each patient's first and last two visits, except for patients with only two visits ($n = 353$). As in Simulation Section \ref{sec.sim}, our analysis included the PROBE and LASSO methods for comparison. However, we did not consider the LMM-LASSO, LASSO+, or PGEE methods due to extensive computational time for $p = 15386 + 38$ predictors. Note that PROBE and LASSO are not methods that account for the presence of random effects. We used the defaults for each method for parameters such as thresholds for convergence and the number of iterations. 

\begin{figure}[t]
\centering
\includegraphics[width=7in]{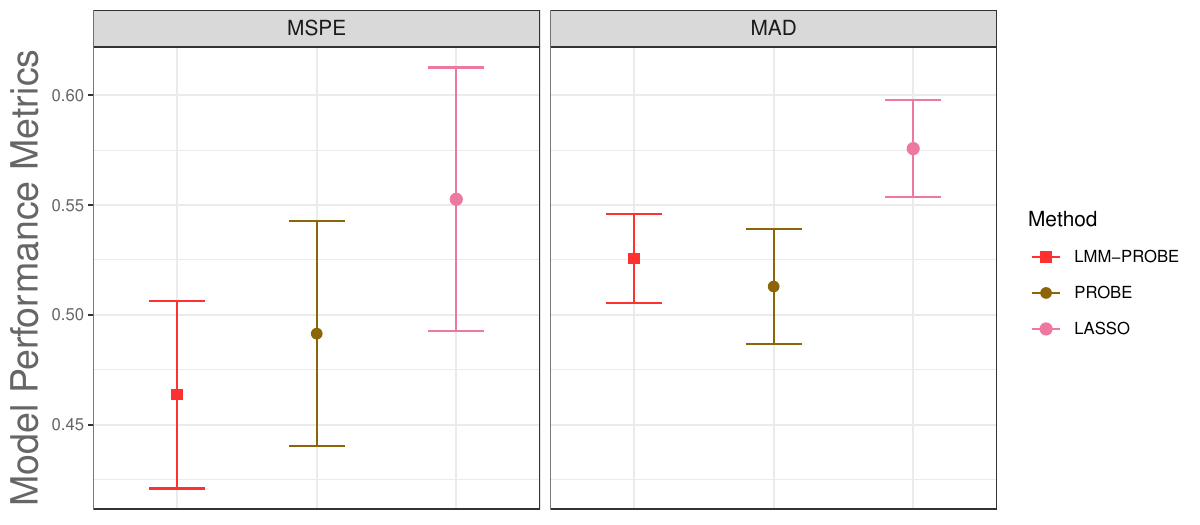} \\
\caption{Mean Squared Predictive Errors (MSPE) and Median Absolute Deviations (MAD) for LMM-PROBE and two comparison methods. Vertical lines represent $\pm$ the standard error of MSPE or MAD, divided by $\sqrt{5}$, based on the number of Cross-Validation folds. \label{fig.mse.example}}
\end{figure}

To evaluate the performance of LMM-PROBE, we used MSPEs and MADs resulting from five-fold CV. In the CV, we balanced patients across the folds, meaning a patient's observations were all in the same fold. At a given iteration of the CV, 80\% of the clusters were in training folds, and 20\% were in the validation-test fold. In models with a random intercept only, we used the validation-test fold to predict the random effects and calculate MSPEs. When fitting models with a random slope for time, we split the validation-test fold into validation (with two of the \textit{time} values for each cluster in the fold, i.e., 1--2) and testing (with the remaining \textit{time} value, i.e., 3) subfolds. We used the validation subfold to obtain the predicted random effects and the testing subfold to calculate MSPEs. For LASSO, we performed an additional five-fold CV using the training folds for parameter tuning.

The initial steps of our analysis indicated that a model with a random intercept (for the patient) and slope (for time) resulted in overfitting (e.g., excessively small random slope variance). Subsequent modeling included only a random intercept term for patients. Figure \ref{fig.mse.example} shows MSPEs as well as MADs. For MSPEs, LMM-PROBE performed best while LASSO ranked last, and the range of MSPE values across the CV folds was wider for PROBE and LASSO compared to LMM-PROBE. Trends in MADs results were similar to those of MSPEs. However, PROBE had the lowest MAD. Supplementary Materials Figure \ref{fig.resvar.example} shows estimates for the within- and between-cluster variances, $\tilde{\sigma}^2$ and $\bG$. For PROBE and LASSO, $\tilde{\sigma}^2$ is the model's residual variance estimate since these methods do not delineate within-unit and between-unit variation. LMM-PROBE had the lowest estimates of $\tilde{\sigma}^2$ across CV folds (on average $0.1$). The average random effect variance $\tilde{\bG}$ estimate was approximately $0.02$ giving an average ICC of 0.15. Overall, LMM-PROBE captured both the residual and random effect variances and provided a better MSPE than PROBE and LASSO, which do not estimate between-cluster variance. 

We examined the average number of predictors selected by each method. For LASSO, a predictor was selected if $\beta_k \ne 0$ while for PROBE and LMM-PROBE, a predictor was selected if $\Tilde{p}_k > 0.5$. LMM-PROBE, PROBE, and LASSO selected 5, 9, and 167 predictors on average, with LMM-PROBE selecting the fewest predictors across CV folds (Supplementary Materials Figure \ref{fig.selected.example}). LMM-PROBE outperformed LASSO and PROBE in this application where it is important to identify strong predictors of IFI6 protein. The most important predictors of IFI6 found by LMM-PROBE were the FBXL19, VAMP2, OR2B2, HIST2H4A, and NIN genes. These genes have been associated with protein or nucleotide interactions and are involved in cell activity changes \citep{Banchereau2016}. A second data analysis example is provided in Supplementary Materials Section \ref{sec.sm.da}.

\section{Discussion}\label{sec.disc}

In this study, we presented a computationally effective and novel method with inferential and predictive capabilities called LMM-PROBE. This new method performs high-dimensional linear mixed modeling even when $p$ reaches the `ultra' high dimensions. An innovative aspect of LMM-PROBE is its partitioning and parameter expansion in the Bayesian context. The mean of the posterior distribution of the regression coefficients is akin to the traditional Bayesian spike-and-slab framework \citep{MitBea88, George1997, Ishwaran2005}, and benefits from a closed-form Gaussian distribution for coefficients of predictors with $\gamma_k \ne 0 $ and a mass point at $0$ for coefficients of predictors with $\gamma_k = 0$. LMM-PROBE leverages these properties and formulates a framework where a) estimation for important parameters requires minimal prior assumptions, b) parameter expansion and partitioning results in computationally effective estimation through an ECM algorithm, which in turn c) allows simultaneous variable selection, coefficient estimation, and prediction in mixed-effects settings.

\textcolor{black}{The literature on (ultra) high-dimensional linear mixed effects regression is still growing. Many of the new proposals are shrinkage methods \citep[e.g.,][]{Bondell2010, Fan2012Aim3, Groll2014, Ibrahim2011, Opoku2021, Sholokhov2024} or Bayesian approaches \citep[e.g.,][]{Kinney2007, Degani2022, Zhou2013}. Our work focused mostly on proposals in the shrinkage area. Well established approaches include the work of \cite{Bondell2010}, \cite{Ibrahim2011}, \cite{Peng2012}, \cite{Fan2012Aim3}, \cite{Chen2003}, as well as \cite{Delattre2020}.} These proposals use shrinkage methods on the fixed and random components via well-known penalty types (LASSO, SCAD, etc.) on the fixed effects and decompositions of $\bG$ to perform random effect selection. However, they do not provide software implementation in \texttt{R}. Methods such as LMM-LASSO \citep{Schelldorfer2011Aim3}, LASSO+ \citep{Rohart2014}, and PGEE \citep{Wang2012} offer software packages, but these implementations suffer from scalability issues due to large matrix operations and nested iterative processes. Figure \ref{fig.time} illustrated the large discrepancies in computation time between LMM-PROBE and LASSO+, LMM-LASSO, and PGEE. As discussed in Section \ref{sec.comp.anal}, this is not unexpected, given the differences in computational complexity of the estimation algorithms.

Our simulations found that compared to LASSO+, LMM-LASSO, and PGEE -- all methods that can be used with non-independent data -- LMM-PROBE demonstrated stronger predictive abilities (lower MSPEs and MADs) for predictions based on either fixed effects only or on fixed and random effects. As expected, LMM-PROBE performed better than variable selection methods for fixed effects only, such as LASSO and PROBE. In conclusion, we proposed a novel approach for (ultra) high-dimensional linear mixed-effect modeling with a software implementation in the \verb|R| package \verb|lmmprobe|. LMM-PROBE is flexible in that it uses empirical Bayes estimators to avoid the specification of hyperpriors for parameters of interest and is computationally efficient through the use of the PX-ECM algorithm that scales linearly in $p$ and $n$. Future research includes extending to generalized linear mixed models, random effects selection, \textcolor{black}{and enabling group fixed effect selection through additional latent $\W_k$ terms.}  

\bigskip
\bigskip

{\bf Disclosure Statement}

The authors report there are no competing interests to declare.

{\bf Word Count}

Word Count: 5293

\pagebreak

\begin{center}
{\large\bf SUPPLEMENTARY MATERIAL}
\end{center}

\begin{description}

\item[Supplementary Material:] \textcolor{black}{The Supplementary Material contains the proofs of Propositions 1 and 2, convergence assessments of LMM-PROBE, and additional results from simulations and
real data analyses.} (.pdf)

\item[R-package:] An R-package \texttt{lmmprobe}, with data from the example, also available on Github \citep{ZgoMcL23}.

\end{description}

\bigskip
\begin{abstract}
\textcolor{black}{In Section \ref{sec.pxecm.em}, we give the proofs of Propositions 1 and 2 from the main text. Section \ref{sec.sm.conv} assesses the convergence of LMM-PROBE.} Section \ref{sec.sm.sim} presents additional simulation results, and Section \ref{sec.sm.da} presents additional results from analyses of real datasets.
\end{abstract}

\newpage
\spacingset{1.75} 

\setcounter{section}{0}
\setcounter{figure}{0}   
\setcounter{table}{0}  
\setcounter{proposition}{0}    
\renewcommand{\thesection}{\Alph{section}}
\renewcommand\thefigure{\thesection.\arabic{figure}}    
\renewcommand\thetable{\thesection.\arabic{table}}

\section{Proofs of Propositions \ref{prop1} and \ref{prop2}}\label{sec.pxecm.em}

\textcolor{black}{We reiterate Propositions \ref{sm.prop1}--\ref{sm.prop2} from the main text and provide proofs for them below.} 
 
\begin{proposition} \label{sm.prop1}
\textcolor{black}{Under a standard EM algorithm, the parameters in the M-step for LMM-PROBE are not always estimable. That is, the maximizer of $Q_{\text{EM}}(\boldsymbol{\eta} \mid \Theta^{(t-1)})$, where $\boldsymbol{\eta} = (\be \ \bomega \ \btau)'$,
is not always unique.
} 
\end{proposition}

\begin{proposition} \label{sm.prop2}
\textcolor{black}{Let
$Q_{\text{CM}}^{M1}(\boldsymbol{\eta} \mid \Theta^{(t-1)})$ and $Q_{\text{CM}}^{M2}(\boldsymbol{\eta} \mid {\Theta}^{(t-1)})$ denote the two M-step quantities in the PX-ECM algorithm for LMM-PROBE.
Assuming no perfect collinearity between $\X_k$ and $\V$ for any $k$, the maximizers of $Q_{\text{CM}}^{M1}(\boldsymbol{\eta} \mid {\Theta}^{(t-1)})$ and $Q_{\text{CM}}^{M2}(\boldsymbol{\eta} \mid {\Theta}^{(t-1)})$ always exist and are unique.} 
\end{proposition}

\begin{proof}[Proof of Proposition \ref{prop1}] \label{sm.prop1proof}
\textcolor{black}{Consider a standard M-step for the regression parameters in (\ref{eq.lmm}) without any additional non-sparse predictors. We add the expanded parameter $\btau$ for $r$-dimensional random effects. For the small $p$ non-sparse situation, \cite{Liuetal98} considered this model (see Section 4.1 therein) and were able to demonstrate superior convergence to the non-expanded parameter version. In this situation, the M-step maximizes}
\begin{equation*}
Q_{\text{EM}}(\boldsymbol{\eta} \mid {\Theta}^{(t-1)}) = - E_{\bgamma \be} \left[ \left\{ \Y - \U(\re)'\mbf{\eta}(\mbf{\gamma})\right\}' 
\left\{ \Y - \U(\re)'\mbf{\eta}(\mbf{\gamma})\right\} \mid {\Theta}^{(t-1)}
\right]
\end{equation*}
\textcolor{black}{where $\U(\re)  = (\X \ \V \ \mbf{\mathcal{V}}'\re)'$ and $\mbf{\eta}(\mbf{\gamma}) = (\bgamma\be \ \bomega \ \btau)$ with $\mbf{\eta}\equiv \mbf{\eta}(\mbf{1})$. For brevity, we drop the $(t-1)$ superscript on all expectations and, without loss of generality, assume $p_k > 0$ (predictor $\X_k$ can be taken out of the optimization when $p_k = 0$). Note that }
\begin{equation*}
E_{\bgamma \re} \left\{ 
 \U(\re)'\mbf{\eta}(\mbf{\gamma})
\right\} = \U(\tilde{\re})'\mbf{\eta}(\mbf{p})  \equiv 
\begin{bmatrix} (\X\odot \mbf{P})' \\ \V' \\ \mbf{\mathcal{V}}' \tilde{\re}
\end{bmatrix}' \mbf{\eta} = 
\tilde\U' \mbf{\eta}, 
\end{equation*} 
\textcolor{black}{where $\odot$ denotes the Hadamard product and $\mbf{P}$ is an $n \times p$ matrix in which each row is the vector $\mbf{p}$, and }
\begin{equation*}
E_{\bgamma \re} \bigg[ \left\{ 
 \U(\re)'\mbf{\eta}(\mbf{\gamma})
\right\}'\left\{ 
 \U(\re)'\mbf{\eta}(\mbf{\gamma})
\right\}\bigg] = (\tilde\U' \mbf{\eta})'\tilde\U' \mbf{\eta} + 
\text{Var} \left\{ 
 \U(\re)'\mbf{\eta}(\mbf{\gamma})
\right\}. 
\end{equation*} 
\textcolor{black}{The variance term reduces to}
\begin{equation*}
\mbf{1}'\begin{bmatrix} (\X^2)\odot \mbf{P}\odot(\mbf{J} - \mbf{P}) \\ \mbf{0}_r \\ \mbf{\mathcal{V}}'(\sigma^{-2} \bPsi)^{-1} \mbf{\mathcal{V}}
\end{bmatrix} \mbf{\eta}^2 = \mbf{\lambda}' \mbf{\eta}^2
\end{equation*} 
\textcolor{black}{where $\X^2 = \X \odot \X$, $\mbf{J}$ is a $(p \times p)$ matrix of ones, $\mbf{0}_r$ is an $r$-dimensional vector of zeros, and $\mbf{\lambda} = \left\{\lambda_1,\ldots,\lambda_p, \mbf{0}_r', \mbf{\mathcal{V}}'(\sigma^{-2} \bPsi)^{-1} \mbf{\mathcal{V}}\right\}'$ where $\lambda_k = p_k(1-p_k)\sum_i X_{ik}^2$ for $k\leq p$.  This yields}
\begin{equation*}
Q_{\text{EM}}(\boldsymbol{\eta} \mid {\Theta}^{(t-1)}) = - \left\{ \left( \Y - \tilde\U' \mbf{\eta} \right)'
 \left( \Y - \tilde\U' \mbf{\eta} \right)
\right\} - \left[ \mbf{\lambda}' \mbf{\eta}^2 \right],
\end{equation*} 
\textcolor{black}{resulting in maximizer}
\begin{equation*}
\hat{\mbf{\eta}} =   \left[ \tilde\U'
 \tilde\U + \mbf{\lambda} \mbf{I}_{p + 2r} \right]^{-1}  \tilde\U' \Y.
\end{equation*} 
\textcolor{black}{
As a result, $\hat{\mbf{\eta}}$ is only estimable when $\sum_k I(p_k = 1) < M + 1 - r$; adding additional non-sparse predictors would lower this limit. Therefore, $Q_{\text{EM}}(\boldsymbol{\eta} \mid {\Theta}^{(t-1)})$ cannot always be maximized, as the M-step is undefined when $\sum_k I(p_k = 1) \ge M + 1 - r$. } 
\end{proof}

\bigskip

\noindent \textcolor{black}{\textbf{Remark 1:} A standard M-step for calculating $\hat{\mbf{\eta}}$ requires the inverse of $\tilde\U' \tilde\U + \mbf{\lambda} \mbf{I}_{p+2r}$ a $(p +2r) \times (p+2r)$ matrix, which is lower-bounded by $\Omega(p^3)$ computational complexity. Note that this matrix is a function of the moments of $\mbf{\gamma}$ and $\re$, which change at every iteration. Consequently, the inversion is required at every iteration, which compounds the computational complexity of the standard M-step. As demonstrated in the simulation studies in the main text, methods that require this type of inversion (e.g., PGEE) scale poorly with $p$ and require computation times that are not feasible for our data analysis.}

\bigskip

\noindent \textcolor{black}{\textbf{Remark 2:} The benefit of the estimate of, say, $\beta_k$ from the standard M-step is that it would adjust for the remaining predictors in the model. This points to the motivation for including $\balpha$, the expanded parameters for $\be$, in the proposed LMM-PROBE method. These expanded parameters adjust for the remaining signal, denoted by $\W_k$ in the main text, while estimating $\beta_k$. Without $\alpha_k$, $\beta_k$ would be estimated using $(\Y - \W_k)$ which is \textit{corrected for} $\W_k$ instead of \textit{adjusted for} $\W_k$. The differences in the estimate of $\beta_k$ with correction versus adjustment can be large when $\X_k$ and $\W_k$ are related.}

\begin{proof}[Proof of Proposition \ref{prop2}] \label{sm.prop2proof}
\textcolor{black}{In the PX-ECM algorithm for LMM-PROBE, the $\ell$th PX-CM partition maximizes the expected posterior of $\bfeta_{\ell}$ (where $\boldsymbol{\eta}$ is a subset of $\bfeta$, $\boldsymbol{\eta} \in \bfeta$), which is a function of $\X_\ell$, $\V$, and the latent terms $\U_\ell$,} 
\begin{equation*}
\hat{\bfeta}_{\ell}^{(M1)} = \mbox{argmax}_{\bfeta_{\ell}} E_{\U_{\ell}}\left\{l(\bfeta_{\ell}|\Y,\U_{\ell}, \mbf{\Gamma}_{\ell})|\Theta_{/\ell}^{(t-1)} \right\} \ \mbox{for} \ \ell=0,1,\ldots,p.
\end{equation*}
\textcolor{black}{With the parameter expansion, this reduces the first M-step (M1) to a straightforward calculation for each partition $\ell$, akin to a $2(r+1)$-dimensional linear regression with maximizer}
\begin{equation*}
\hat{\mbf{\xi}}^{(M1)}_{\ell}=\left\{(\Z_{\ell}'\Z_{\ell})^{(t-1)}\right\}^{-1}\Z^{(t-1)'}_{\ell}\Y,
\end{equation*}
\textcolor{black}{where $\Z_{\ell}$ is defined in Section \ref{sec.pep} of the main text. This maximizer is fully identifiable under the PX-ECM framework as long as there is not perfect collinearity between $\X_k$ and $\V$ for all $k$. As a result, $Q_{\text{CM}}^{M1}(\boldsymbol{\eta} \mid {\Theta}^{(t-1)})$ can be maximized. The second M-step M2 updates $\hat{\mbf{\xi}}^{(M2)}_{0}$, a subset of $\boldsymbol{\eta}$, conditional on the first M-step M1 and E-steps. This yields the same proof that the maximizer is identifiable and exists. Therefore, the maximizers of $Q_{\text{CM}}^{M1}(\boldsymbol{\eta} \mid {\Theta}^{(t-1)})$ and $Q_{\text{CM}}^{M2}(\boldsymbol{\eta} \mid {\Theta}^{(t-1)})$ both exist and are unique.}
\end{proof}

\section{Convergence Assessment}\label{sec.sm.conv}

\textcolor{black}{This Section provides additional results and discussion with regards to convergence. Our implementation of LMM-PROBE is an `all-at-once' optimization, where the parameter value updates at iteration $t$ are performed accounting for updates from the previous iteration $t-1$ \citep[analogous to the Jacobi least-squares optimization approach,][]{Mas95}. In previous work, we have found that the `all-at-once' optimization was less sensitive to the updating order and more efficiently maximized the likelihood than the `one-at-a-time' approach \citep{McLain2022}. In the `one-at-a-time' approach, parameter values are sequentially updated, with each update accounting for all previous updates within current iteration $t$ \citep[analogously to Gauss-Seidel least-squares optimization,][]{Maetal15}. The `one-at-a-time' approach fulfills the monotonicity and space-filling properties of EM and ECM algorithms \citep{Meng1993}. Our `all-at-once' implementation of LMM-PROBE is less prone to getting stuck in local maxima (in practice) compared to a `one-at-a-time' approach but does not provide convergence guarantees.}

\textcolor{black}{Various studies have demonstrated that the PX-EM algorithm and some parameterizations of ECM have better global rates of convergence than a standard EM algorithm \citep{Liuetal98, Men94, Sexton2000}. However, for LMM-PROBE, the parameters in the complete data model of the standard EM algorithm are nonidentifiable when $p \gg n$ (Proposition 1), and thus the M-step is often undefined. Therefore, we cannot compare the theoretical rate of convergence of a PX multi-cycle ECM as in LMM-PROBE to that of a standard EM. Instead, we assess empirical evidence of convergence. For eight simulation settings chosen to illustrate all simulation factors, Figure \ref{fig.conv.iter} shows that the log-likelihood sharply increased in early iterations and flattened quickly in a small number of iterations, indicating a rapid rate of convergence. The empirical convergence assessment was similar for other simulation settings (omitted from figures).}

\begin{figure}[ht]
\centering
\includegraphics[width=6.5in]{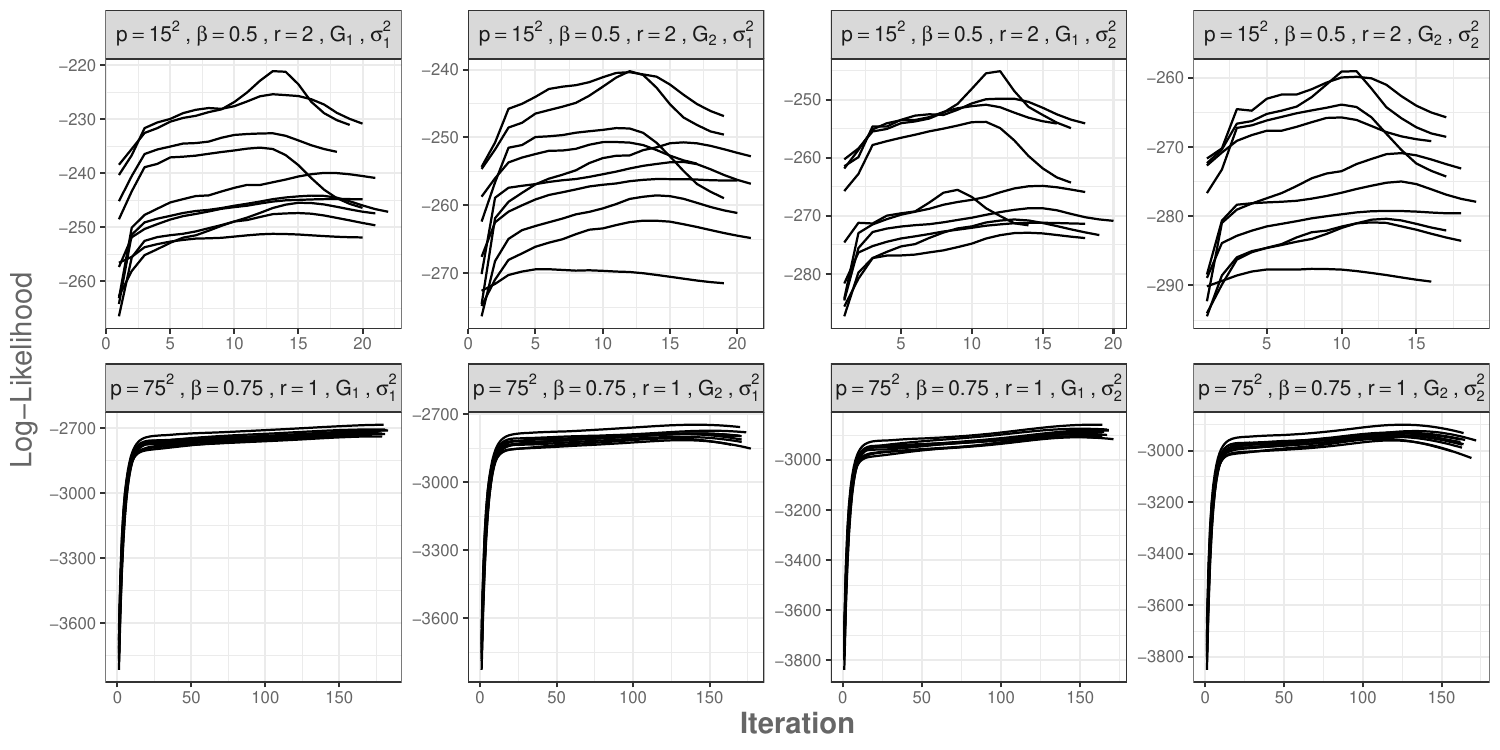} \\
\caption{Log-Likelihood for LMM-PROBE over iterations, across various simulation settings, including $r$, $\sigma^2$, $\bG$, and $\mbf{\beta}$ values, when $p = (15^2, 75^2)$ and $\pi = 0.1$. Within each simulation scenario, lines represent simulation repetitions (10 repetitions). \label{fig.conv.iter}}
\end{figure}

\textcolor{black}{To empirically assess the maximization of the conditional likelihood, we examined $- \frac{n}{2}\log(\tilde{\sigma}^2) - \frac{1}{2\tilde{\sigma}^2}\sum_i \left[\Y_{i} - \left\{ \X_{i} (\tilde{\alpha}_0\tilde{\mbf{p}} \tilde \be)  + \V_{i} \tilde{\bomega}_0  +  \V_{i} \tilde \re_i \right\} \right]^2$ for all methods. Figure \ref{fig.loglik} shows that LMM-PROBE maximized the conditional likelihood better than LASSO and PROBE in the settings displayed. Additionally, LMM-PROBE maximized the conditional likelihood better than LASSO+ in all settings except two, and better than LMM-LASSO when $\be = 0.5$ and variances were higher. When $\pi = 0.05$, LMM-PROBE maximized the conditional likelihood most effectively in all settings (figure omitted). Generally, methods for LMMs maximized the conditional likelihood better than LASSO and PROBE, which are not designed for LMMs.}

\begin{figure}[ht]
\centering
\includegraphics[width=7in]{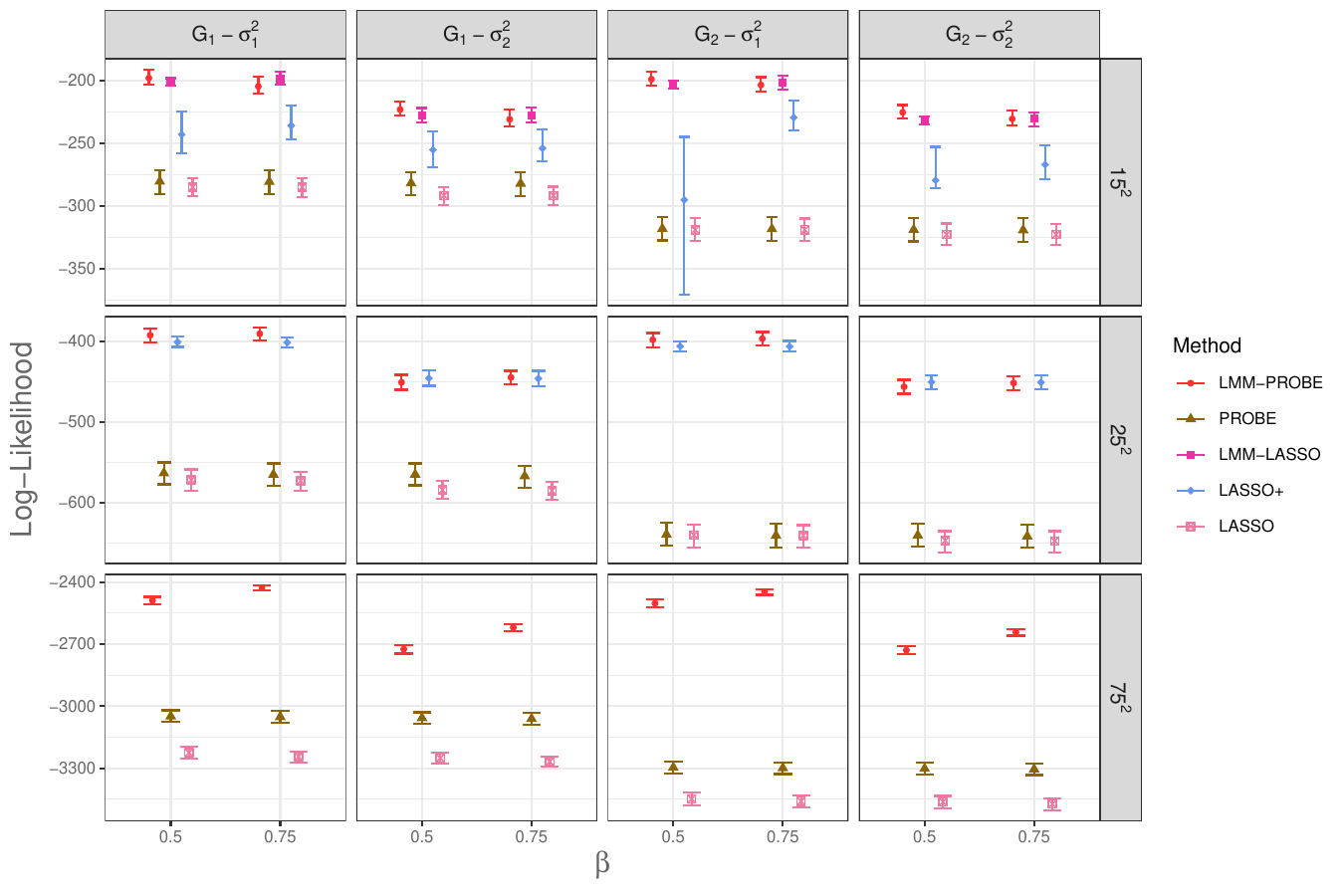} \\
\caption{Log-Likelihood for LMM-PROBE and four comparison methods (excluding PGEE) across various simulation settings, including $p$, $\sigma^2$, $\bG$, and $\mbf{\beta}$ values, when $r = 2$ and $\pi = 0.1$. Vertical lines display the interquartile range of the Log-Likehood values across simulation iterations. Comparison methods LMM-LASSO and LASSO+ are methods for linear mixed models or repeated measures, while LASSO and PROBE are methods for linear models. 
\label{fig.loglik}}
\end{figure}

\section{Additional Simulation Results}\label{sec.sm.sim}

This Section provides additional simulation results. Figure \ref{fig.mspe.pgee} corresponds to Figure \ref{fig.mse.simul} in Section \ref{sec.sim} of the main text, with the additional method of penalized generalized estimating equations (PGEE) for high-dimensional longitudinal data. The Mean Squared Predictive Error (MSPE) was the highest for this additional method. Figure \ref{fig.mspe.1re} also mirrors Figure \ref{fig.mse.simul} from Section \ref{sec.sim} in the main text, but shows MSPEs for the simulation settings with one random effect ($r = 1$). Figure \ref{fig.mad.1re} shows Median Absolute Deviations (MADs) for simulation settings when $r=1$ and $\pi = 0.1$. In all settings except one, MADs were lowest for LMM-PROBE. \textcolor{black}{Figures \ref{fig.varselect625} and \ref{fig.varselect5625} show sensitivity, specificity, and the Matthews Correlation Coefficient \citep[MCC, ][]{Matthews1975} when $p = 25^2$ and $p = 75^2$, respectively. Figure \ref{fig.varselect625} shows that when the effect size of signals was higher ($\be = 0.75$), LMM-PROBE had a similar sensitivity as PROBE, and outperformed other methods. LMM-PROBE had the highest specificity for all settings. In Figure \ref{fig.varselect5625}, LMM-PROBE had higher sensitivity when $\be = 0.75$ as well as high sensitivity in all settings except one.}

\begin{figure}[ht]
\centering
\includegraphics[width=7in]{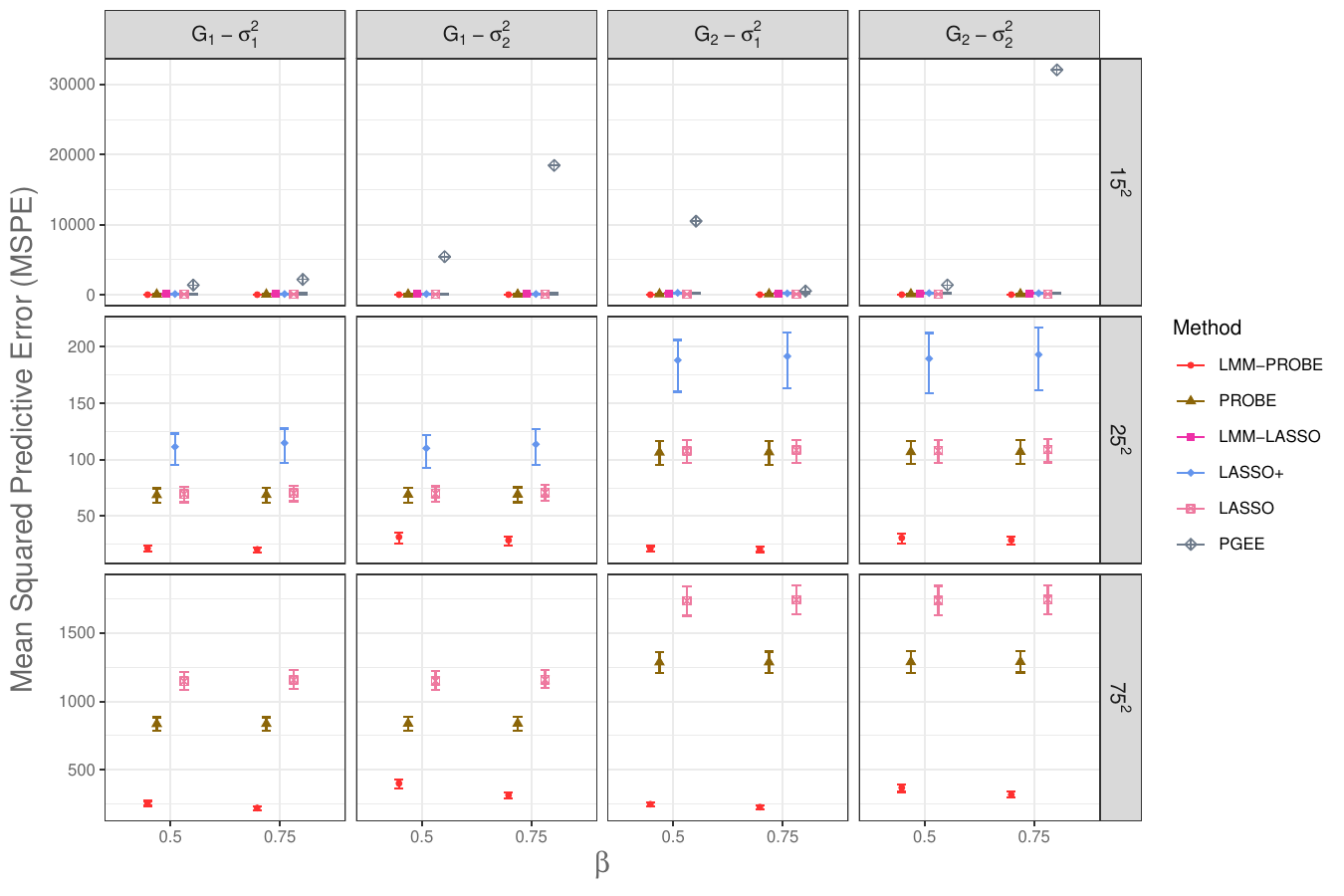} \\
\caption{Mean Squared Predictive Errors (MSPE) for LMM-PROBE and five comparison methods (including PGEE) across various simulation settings, including $p$, $\sigma^2$, $\bG$, and $\mbf{\beta}$ values, when $r = 2$ and $\pi = 0.1$. The MSPEs are based on both fixed and random effects. Vertical lines display the interquartile range of the MSPEs. Comparison methods LMM-LASSO, LASSO+, and PGEE are methods for linear mixed models or repeated measures, while LASSO and PROBE are methods for linear models. \label{fig.mspe.pgee}}
\end{figure}

\begin{figure}[ht]
\centering
\includegraphics[width=7in]{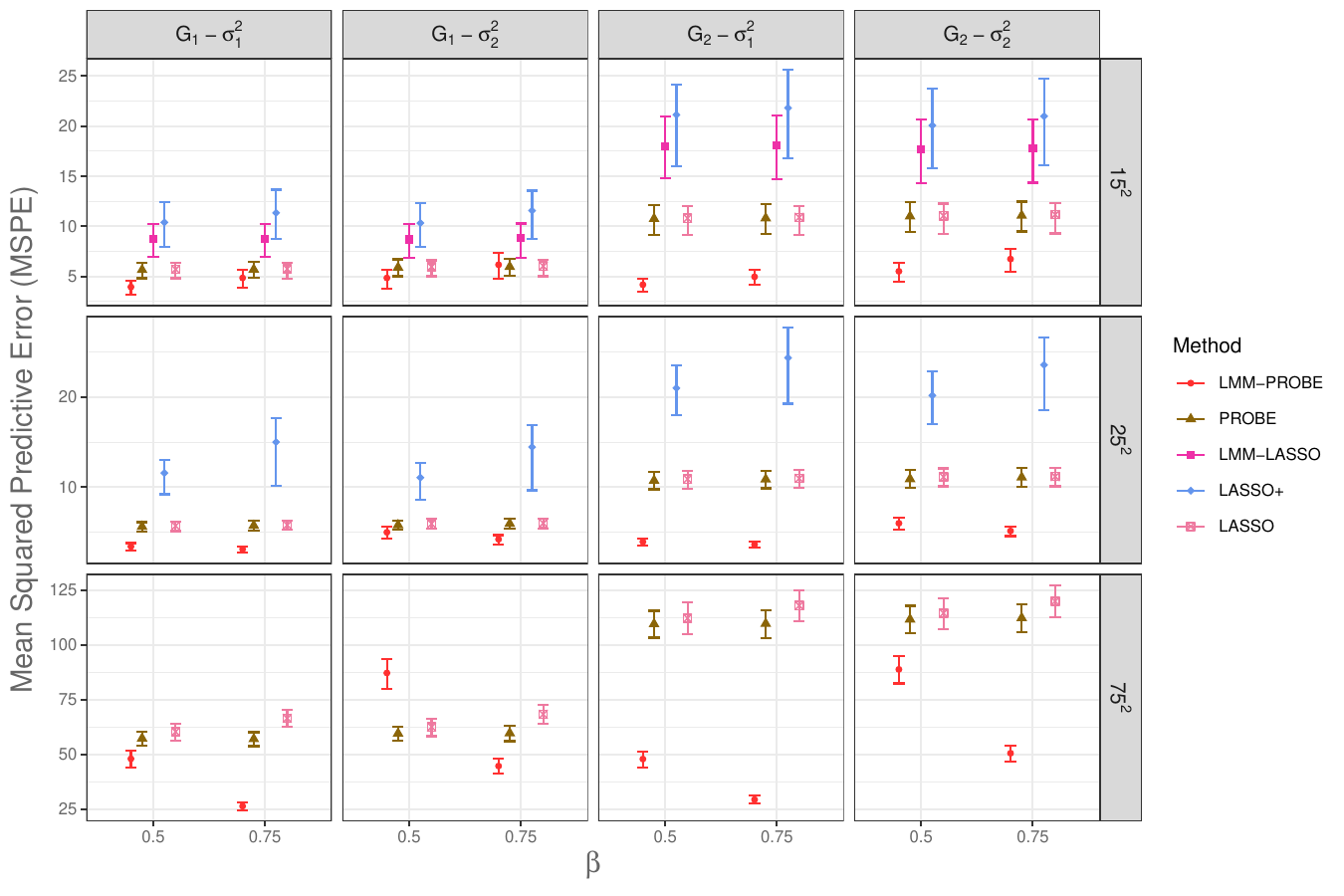} \\
\caption{Mean Squared Predictive Errors (MSPE) for LMM-PROBE and four comparison methods across various simulation settings, including $p$, $\sigma^2$, $\bG$, and $\mbf{\beta}$ values, when $r = 1$ and $\pi = 0.1$. The MSPEs are based on both fixed and random effects. Vertical lines display the interquartile range of the MSPEs. Comparison methods LMM-LASSO and LASSO+ are methods for linear mixed models, while LASSO and PROBE are methods for linear models. \label{fig.mspe.1re}}
\end{figure}

\begin{figure}[ht]
\centering
\includegraphics[width=7in]{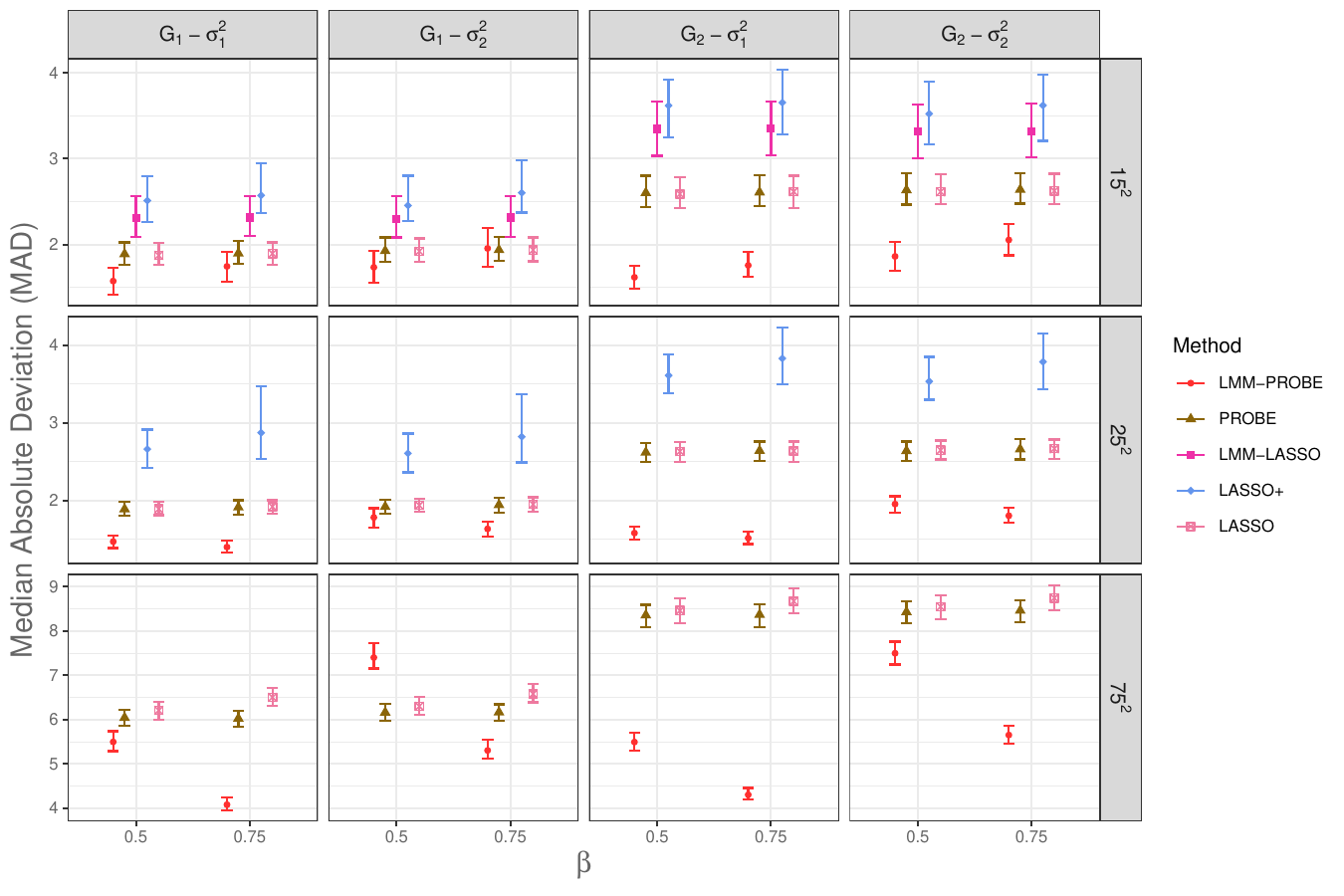} \\
\caption{Median Absolute Deviations (MAD) for LMM-PROBE and four comparison methods across various simulation settings, including $p$, $\sigma^2$, $\bG$, and $\mbf{\beta}$ values, when $r = 1$ and $\pi = 0.1$. The MADs are based on both fixed and random effects. Vertical lines display the interquartile range of the MADs. Comparison methods LMM-LASSO and LASSO+ are methods for linear mixed models, while LASSO and PROBE are methods for linear models. \label{fig.mad.1re}}
\end{figure}

\begin{figure}[ht]
\centering
\includegraphics[width=7in]{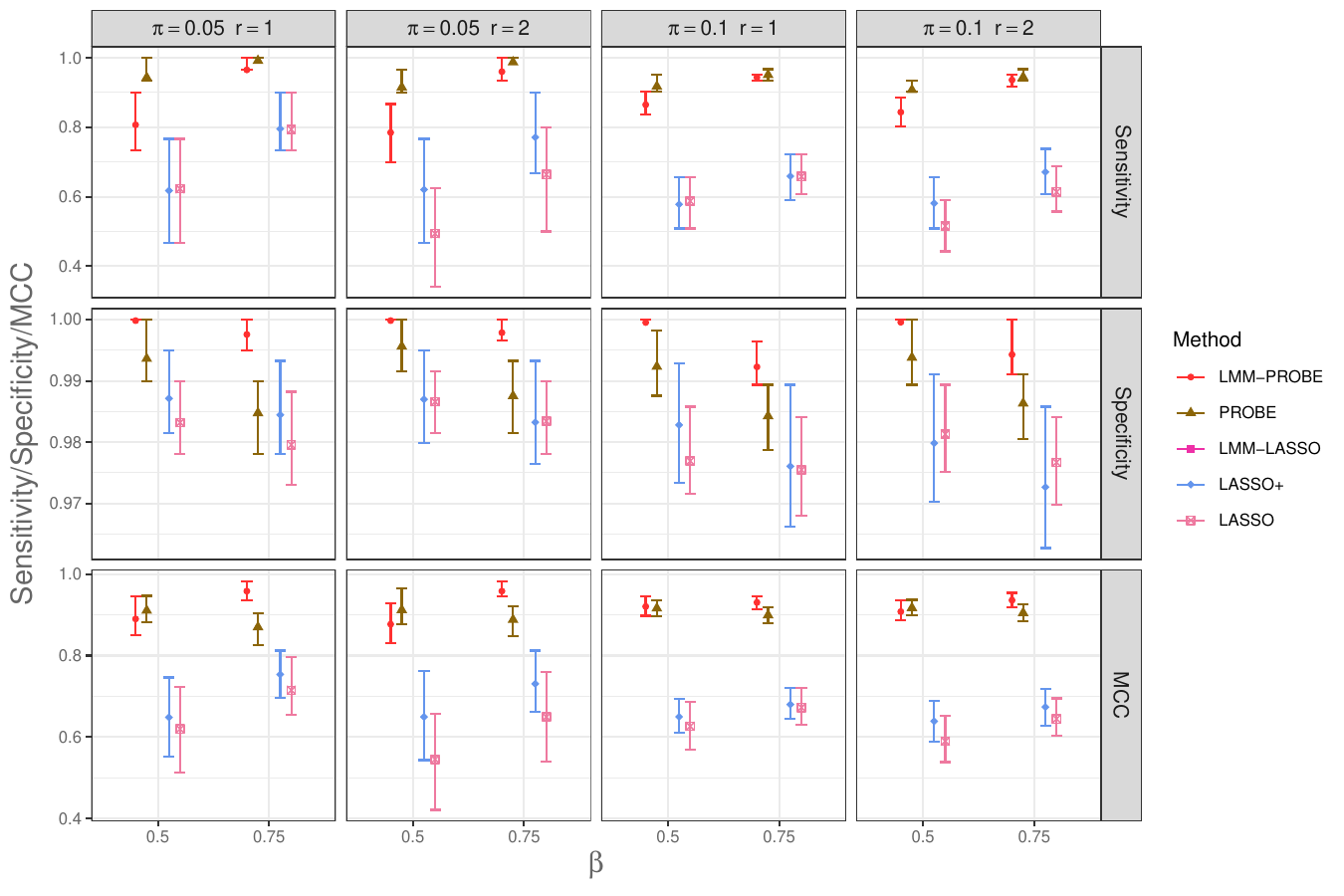} \\
\caption{Sensitivity, Specificity, and the Matthews Correlation Coefficient (MCC) for LMM-PROBE and four comparison methods across various simulation settings, including $\pi$, $r$, and $\be$ values, when $\sigma^2 = \sigma^2_1$, $\bG = \bG_1$, and $p = 25^2$. Vertical lines display the interquartile range of the sensitivity, specificity, and MCC. Comparison methods LMM-LASSO and LASSO+ are methods for linear mixed models, while LASSO and PROBE are methods for linear models. \label{fig.varselect625}}
\end{figure} 

\begin{figure}[ht]
\centering
\includegraphics[width=7in]{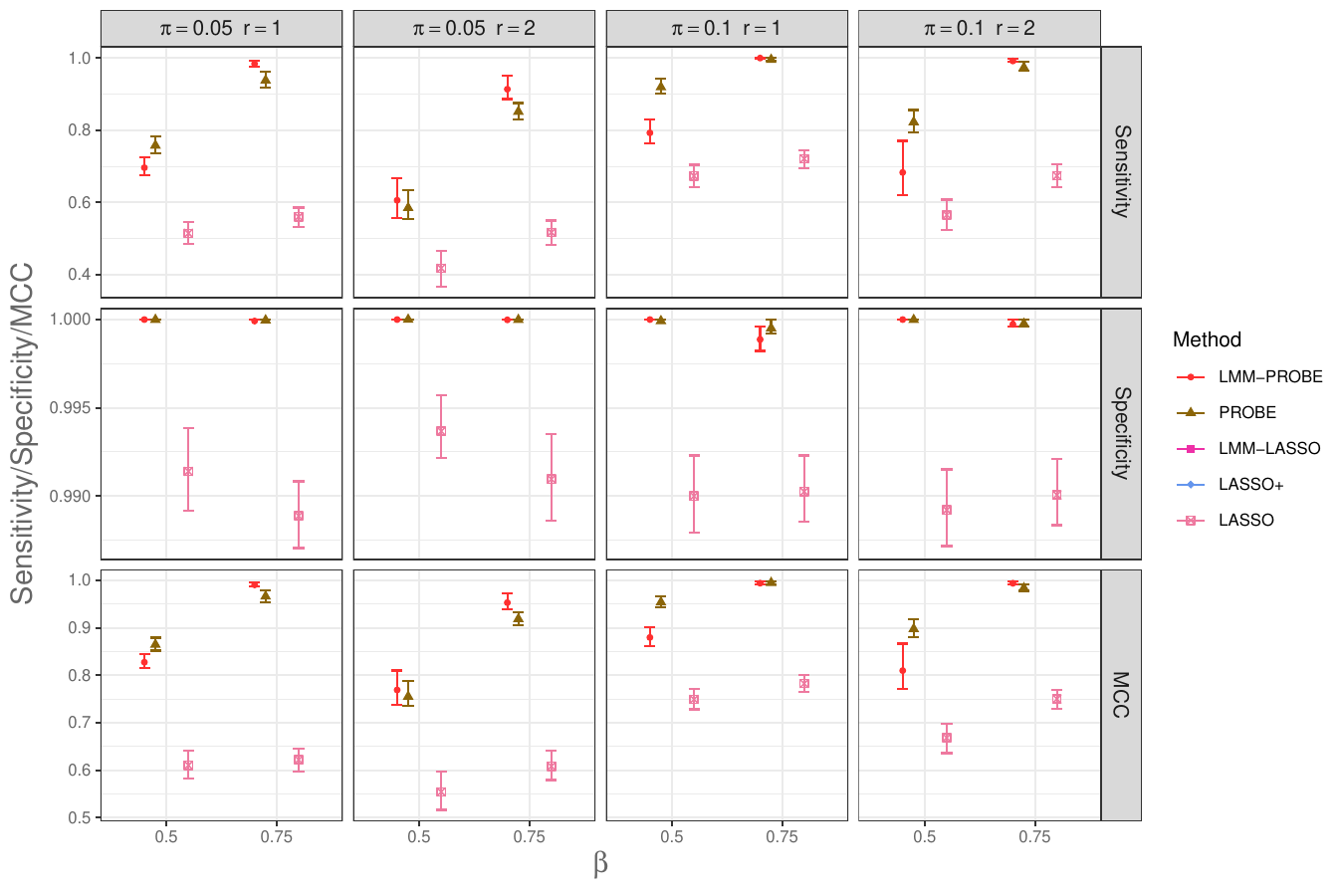} \\
\caption{Sensitivity, Specificity, and the Matthews Correlation Coefficient (MCC) for LMM-PROBE and two comparison methods across various simulation settings, including $\pi$, $r$, and $\be$ values, when $\sigma^2 = \sigma^2_1$, $\bG = \bG_1$, and $p = 75^2$. Vertical lines display the interquartile range of the sensitivity, specificity, and MCC. Comparison methods LASSO and PROBE are methods for linear models. \label{fig.varselect5625}}
\end{figure} 

\section{Additional Data Analysis Results}\label{sec.sm.da}

In this Section, we examine additional results from the data analysis on pediatric systemic lupus erythematosus (SLE) and provide a second data analysis on the popular riboflavin dataset \citep{Buhlmann2014}. For the SLE analysis, Figure \ref{fig.resvar.example} shows that LMM-PROBE had the lowest estimates of $\tilde{\sigma}^2$ across Cross-Validation (CV) folds, compared to PROBE and LASSO. LASSO showed the largest variation in its $\tilde{\sigma}^2$ estimates across CV folds. The average random effect variance $\tilde{\bG}$ estimate was approximately $0.02$ giving an average ICC of 0.15. Figure \ref{fig.selected.example} shows the average number of selected predictors for each method, across CV folds. Computation time for the three methods varied slightly. The methods for traditional high-dimensional linear regression (LASSO, PROBE) required around nine and 48 seconds per fold, respectively, while LMM-PROBE required 69 seconds on average.

\begin{figure}[ht]
\centering
\includegraphics[width=6in]{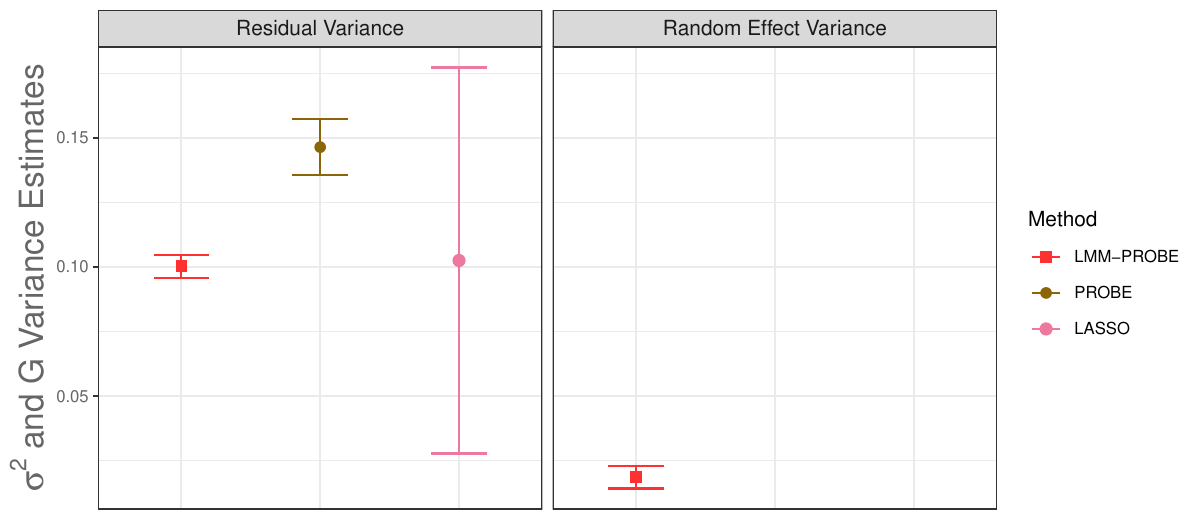} \\
\caption{Average residual ($\tilde{\sigma}^2$) and random effect ($\bG$) variance estimates for LMM-PROBE and two comparison methods. Vertical lines represent $\pm$ the standard error of $\tilde{\sigma}^2$ or $\bG$, divided by $\sqrt{5}$, based on the number of Cross-Validation folds. 
\label{fig.resvar.example}}
\end{figure}

\begin{figure}[ht]
\centering
\includegraphics[width=6in]{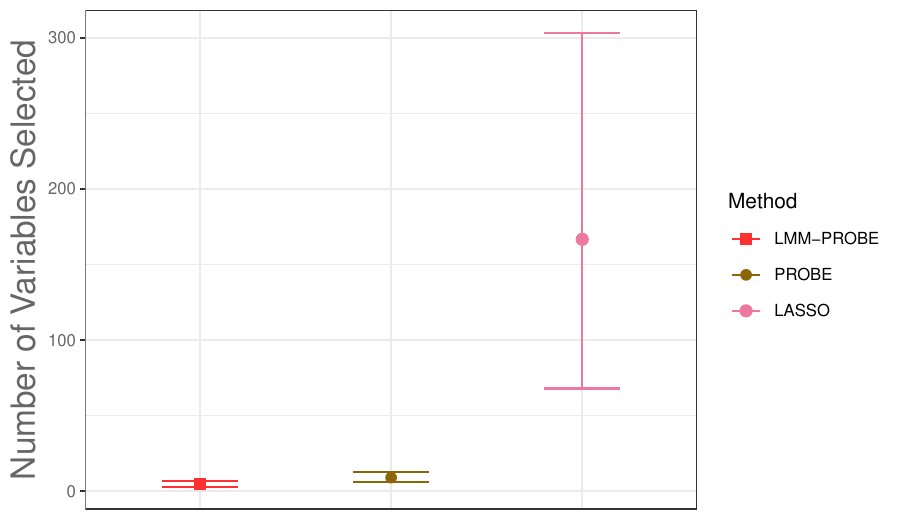} \\
\caption{Average number of predictors selected for LMM-PROBE and two comparison methods. Vertical lines display the range across the Cross-Validation folds.
\label{fig.selected.example}}
\end{figure}

\textcolor{black}{We examined residuals for each CV fold. Figure \ref{fig.condres.obs} shows conditional residuals \citep{Nobre2007} plotted against observations for each CV fold. These residuals show that the homogeneous variance assumption is fulfilled. Quantile-Quantile (Q-Q) plots in Figure \ref{fig.condres.qq} show that the normality assumption is not violated. Figure \ref{fig.re.sub} shows the predicted random intercepts plotted against subjects, where no subject had an outlying predicted random intercept. Finally, the Q-Q plots in Figure \ref{fig.re.qq} show that the random effect normality assumption is not violated either.}

\begin{figure}[ht]
\centering
\includegraphics[width=6in]{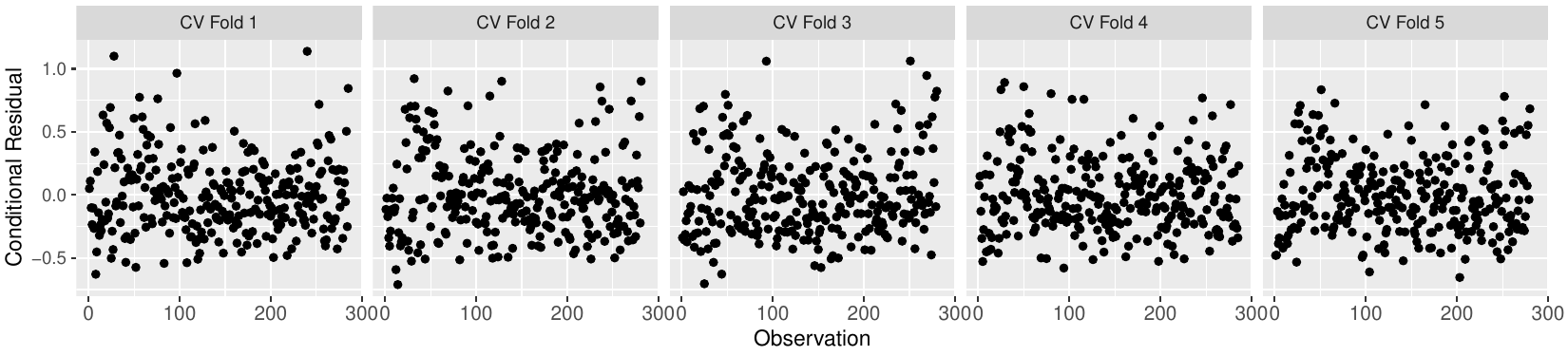} \\
\caption{Conditional residuals plotted against observations for each Cross-Validation fold.
\label{fig.condres.obs}}
\end{figure}

\begin{figure}[ht]
\centering
\includegraphics[width=6in]{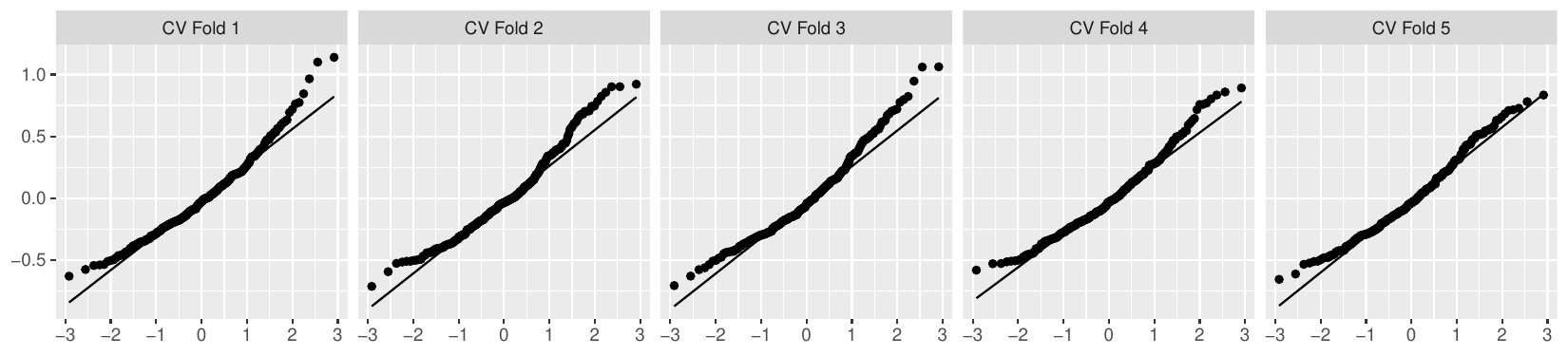} \\
\caption{Quantile-Quantile plots of conditional residuals for each Cross-Validation fold.
\label{fig.condres.qq}}
\end{figure}

\begin{figure}[ht]
\centering
\includegraphics[width=6in]{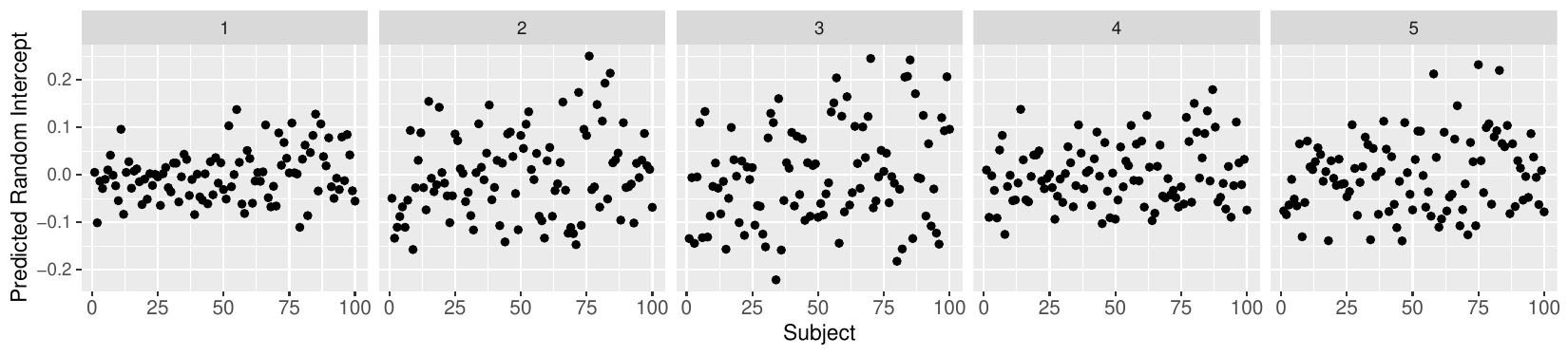} \\
\caption{Predicted random intercepts plotted against subjects for each Cross-Validation fold.
\label{fig.re.sub}}
\end{figure}

\begin{figure}[ht]
\centering
\includegraphics[width=6in]{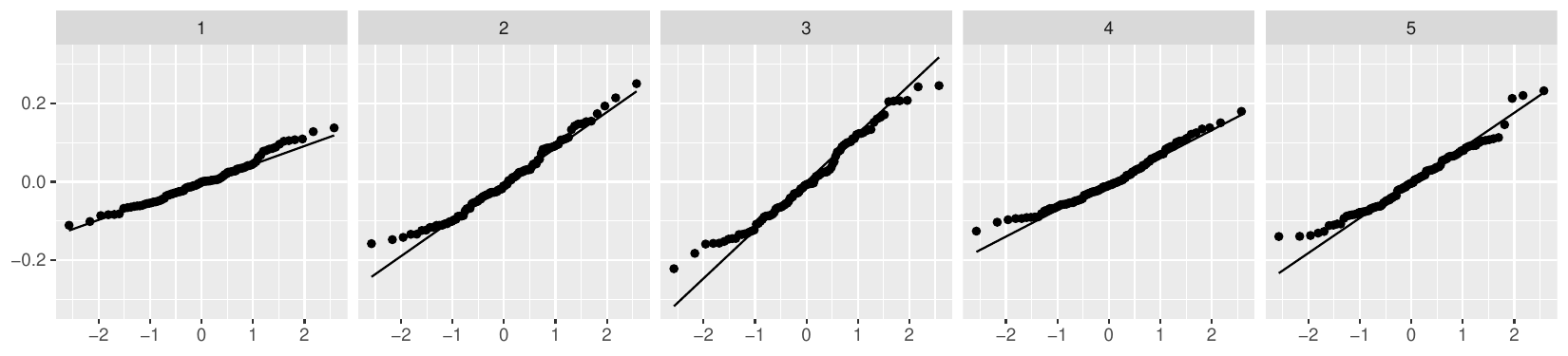} \\
\caption{Quantile-Quantile plots of predicted random intercepts for each Cross-Validation fold.
\label{fig.re.qq}}
\end{figure}

We now showcase the performance and characteristics of the LMM-PROBE method on the riboflavin (vitamin B$_2$) dataset \citep{Buhlmann2014}, a popular high-throughput genomic dataset about riboflavin production using a genetically modified bacterium. Riboflavin is grown in recombinant \textit{Bacillus subtilis}, with the goal of maintaining the production rate of the riboflavin over time. Riboflavin grows across multiple generations of the bacterium for a different duration, leading to a longitudinal design with repeated observations. The dataset contains a log-transformed riboflavin production rate, which we use as the outcome, as well as $4088$ predictor variables, each consisting of gene expression levels on the log scale. There are $28$ clusters in the dataset, measured over time, with the number of observations per cluster ranging between $2$ and $6$. In total, there are $111$ observations. The original aim of the research that generated this dataset was to identify the genes that impacted riboflavin production. To analyze this data, we used a linear mixed-effects model. We considered all $4088$ gene expression predictors as sparse fixed effects and considered an intercept as well as a time predictor as both non-sparse fixed and random effects. 

To evaluate the performance of LMM-PROBE, we used MSPEs resulting from five-fold CV. In the CV, we balanced clusters across the folds, meaning that a cluster's observations were all in the same fold. At a given iteration of the CV, 80\% of the clusters were in training folds, and 20\% were in the validation-test fold. The validation-test fold was further split into validation (with earlier \textit{time} values for each cluster in the fold, i.e., 1--2) and testing (with subsequent \textit{time} values, i.e., 3--6) subfolds. The testing subfold had two observations from each cluster with four or more measurements and one observation otherwise. We used the validation subfold to obtain the predicted random effects and the testing subfold to calculate MSPEs. To generate random effect predictions for the validation subfold, we used the MAP estimate of LMM-PROBE. To tune the penalty parameter in LASSO, we performed an additional five-fold CV using the training set. 

\begin{figure}[t]
\centering
\includegraphics[width=6in]{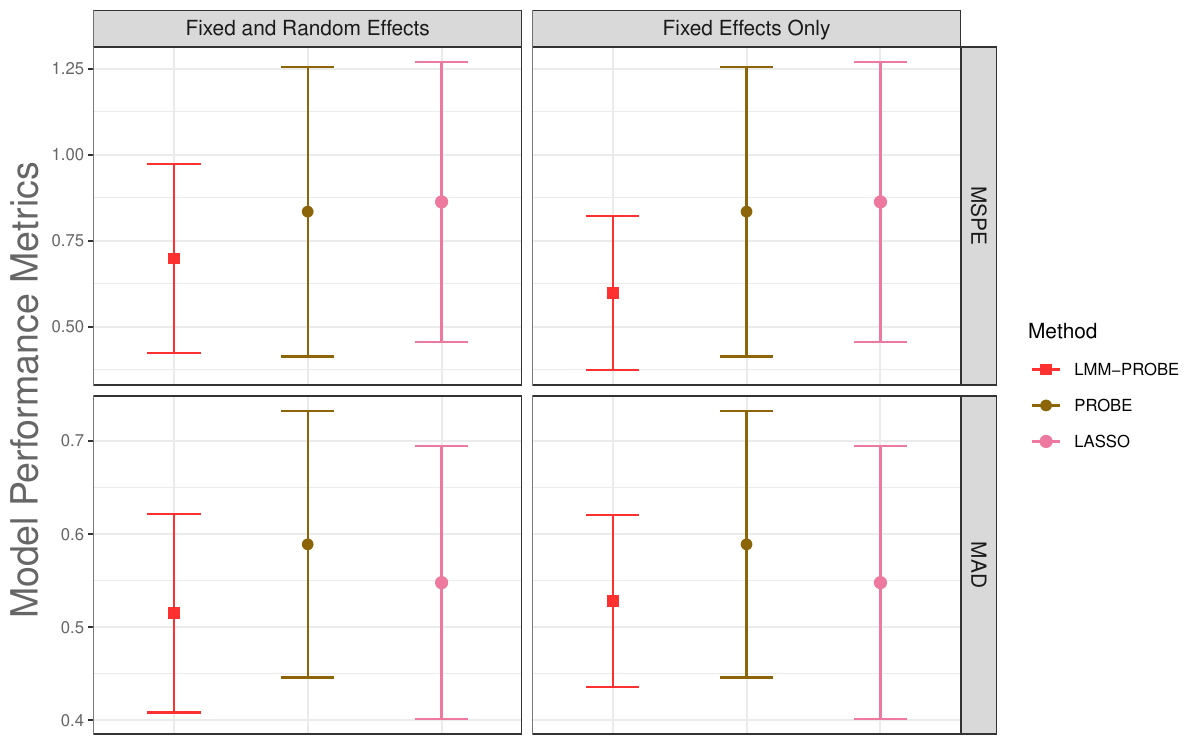} \\
\caption{Mean Squared Predictive Errors (MSPE) and Median Absolute Deviations (MAD) for LMM-PROBE and two comparison methods, based on both random and fixed effects and fixed effects only. 
Vertical lines represent $\pm$ the standard error of MSPE or MAD, divided by $\sqrt{5}$, based on the number of Cross-Validation folds. \label{fig.mse.example2}}
\end{figure}

We calculated two types of MSPEs. The first type of MSPE is the difference between the true outcome value and prediction obtained based on both the fixed and random effect estimates. Specifically, using estimates $\tilde{\bG}$ and $\tilde{\sigma}^2$ from the training data, we obtained predicted random effects $\tilde{\re}_i$ using the validation data and finally obtained complete predictions for future time points using the testing data. The second type of MSPE is the difference between the true outcome value and predictions obtained only based on the fixed effect estimates. Figure \ref{fig.mse.example2} shows the two types of MSPEs as well as MADs. For MSPEs based on both fixed and random effect estimates, LMM-PROBE had the lowest MSPEs compared to PROBE and LASSO. 
When examining MSPEs based on fixed effect estimates only, LMM-PROBE performed best as well. Generally, the range of MSPE values across the CV folds was wider for LASSO and PROBE, and narrower for LMM-PROBE. 
Trends in MADs results were overwhelmingly similar to those of MSPEs. 

\begin{figure}[t]
\centering
\includegraphics[width=6in]{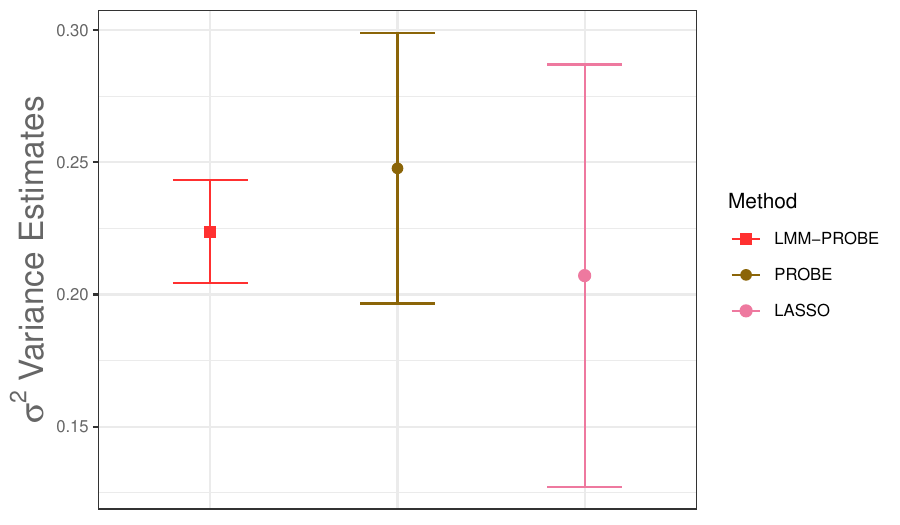} \\
\caption{
Average residual ($\tilde{\sigma}^2$) variance estimates for LMM-PROBE and two comparison methods. Vertical lines represent $\pm$ the standard error of $\tilde{\sigma}^2$ divided by $\sqrt{5}$, based on the number of Cross-Validation folds. \label{fig.sig.example2}}
\end{figure}

Figure \ref{fig.sig.example2} shows estimates for within-sample variation, $\tilde{\sigma}^2$. For PROBE and LASSO, $\tilde{\sigma}^2$ is the model residual error estimates since these methods do not delineate between within-unit and between-unit variation. LMM-PROBE had a mid-range estimate for $\tilde{\sigma}^2$ with the lowest variability across CV fold. LASSO had the lowest average $\tilde{\sigma}^2$ estimate but varied markedly across CV folds. Figure \ref{fig.selected.example2} shows the average number of predictors selected by LMM-PROBE (1), PROBE (1), and LASSO (33), with LMM-PROBE selecting the fewest predictors across CV folds. 
Prior to choosing the final model described above with the LMM-PROBE, LASSO, and PROBE approaches, we examined additional methods (LASSO+, LASSO+EN, and BRMS) but did not continue with them due to computation time. Figure \ref{fig.time.example2} shows the average computation time in seconds for one CV fold of an intermediate model in our data analysis. The methods for traditional high-dimensional linear regression (LASSO, PROBE) required around one second per fold. The remaining methods required between 200 and 6,000 seconds, on average, per fold. Overall, the total computation time for the intermediate model was 1,026 minutes for BRMS, 202 minutes for LASSO+EN, 51 minutes for LASSO+, 2 minutes for LMM-PROBE, 31 seconds for PROBE, and 18 seconds for LASSO.  

\begin{figure}[ht]
\centering
\includegraphics[width=6in]{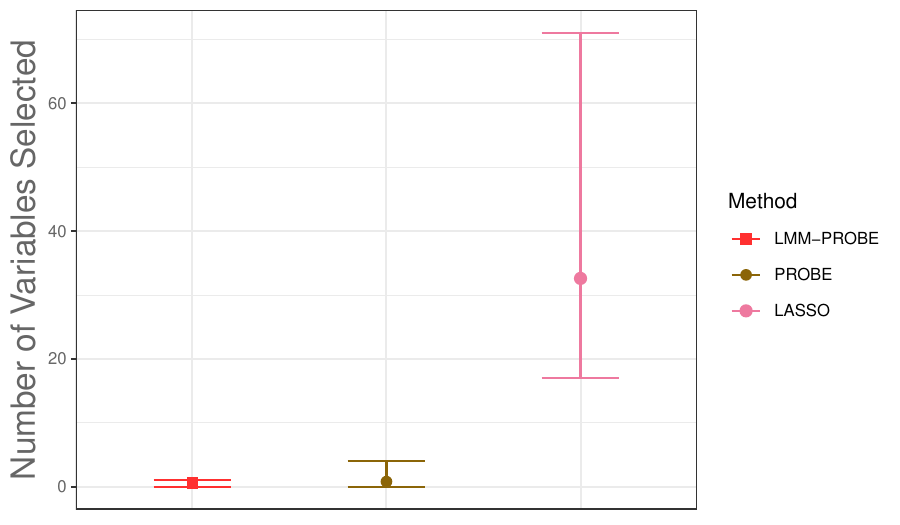} \\
\caption{Average number of predictors selected for LMM-PROBE and two comparison methods. Vertical lines display the range across the Cross-Validation folds. \label{fig.selected.example2}}
\end{figure}

\begin{figure}[ht]
\centering
\includegraphics[width=6in]{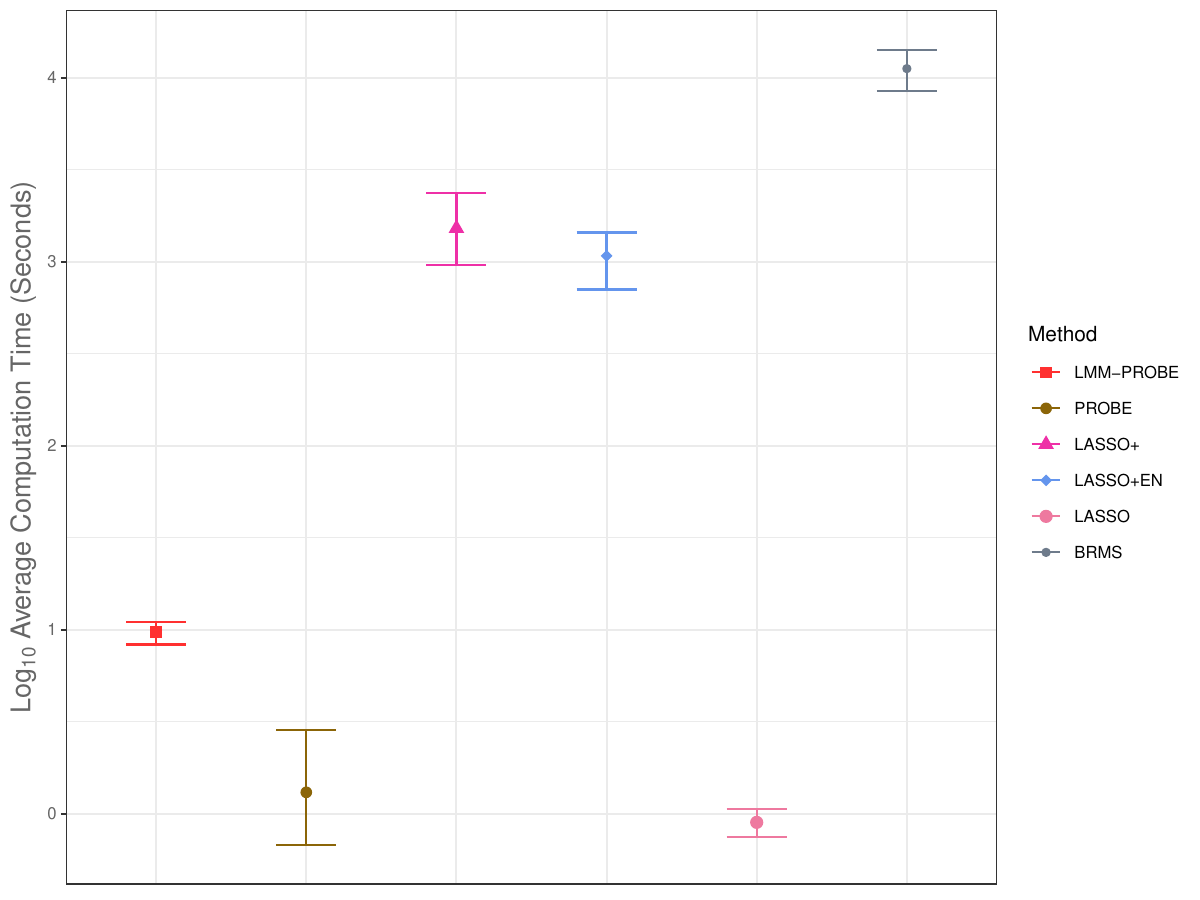} \\
\caption{Average computation time per Cross-Validation fold (in seconds) for LMM-PROBE and five comparison methods, on the log$_{10}$ scale.  
\label{fig.time.example2}}
\end{figure}

\clearpage

\bibliographystyle{chicago}

\bibliography{Dissertation_References}
\end{document}